\documentclass[twoside,11pt]{article}

\usepackage{blindtext}
\usepackage[preprint]{jmlr2e}

\usepackage{amsmath,amssymb,amsfonts}

\usepackage{amsmath,amsfonts,bm}

\def\eqref#1{equation~\ref{#1}}

\def\1{\bm{1}}

\DeclareMathAlphabet{\mathsfit}{\encodingdefault}{\sfdefault}{m}{sl}
\SetMathAlphabet{\mathsfit}{bold}{\encodingdefault}{\sfdefault}{bx}{n}

\makeatletter
\newcommand\notsotiny{\@setfontsize\notsotiny\@vipt\@viipt}
\makeatother

\usepackage{paralist}
\usepackage{algorithmic}
\usepackage{graphicx}
\usepackage{textcomp}
\usepackage{xcolor}
\usepackage{thmtools}
\usepackage{xspace}

\usepackage[utf8]{inputenc}

\usepackage{url}
\usepackage{csquotes}
\usepackage[inline]{enumitem}
\usepackage{nicefrac}
\usepackage[capitalize,noabbrev]{cleveref}
\usepackage{bm}
\usepackage{stmaryrd}

\usepackage{lastpage}

\title{From learnable objects to learnable random objects}

\date{}

\author{\name Aaron Anderson \addr University of Pennsylvania \\
\AND
\name Michael Benedikt
\addr University of Oxford
}

\newcommand{\branch}{\mathrm{B}}
\newcommand{\realvtree}{\tau}
\newcommand{\branchlabelling}{\Lambda}
\newenvironment{myexmp}{\refstepcounter{myexmp}\par\medskip
\noindent\textbf{Example~\themyexmp.}}{\null\hfill \smallskip}
\newcounter{myexmp}[section]
\renewcommand{\themyexmp}{\thesection.\arabic{myexmp}}

\newcommand{\compatparamrv}{\mathrm{CompatParamRV}}
\newcommand{\compatrangerv}{\mathrm{CompatRangeRV}}
\newcommand{\hypoclass}{{\mathcal H}}
\newcommand{\hypo}{h} 
\newcommand{\concept}{{\mathcal C}}
\newcommand{\classfamily}{{\mathcal F}}
\newcommand{\expclass}{{\mathbb E}}

\newcommand{\randvar}{{\mathcal{RV}}}

\newcommand{\measureclass}[1]{{\mathcal Distr}_{#1}}
\newcommand{\dualmeasureclass}[1]{{\mathcal DualDistr}_{#1}}

\newcommand{\expectedloss}{\mathrm{ExpLoss}}
\newcommand{\bestexpectedloss}{\mathrm{BestExpLoss}}
\newcommand{\pshatter}{\mathrm{FatSHDim}}
\newcommand{\fatshatter}{\pshatter}
\newcommand{\seqfatshatter}{\pshatter^{Seq}}

\newcommand{\amodel}{\mathfrak M}

\newcommand{\glivcant}{GC}
\newcommand{\conceptclass}{{\mathcal C}}

\newtheorem{fact}{Fact}

\newcommand{\paramrandomized}{\mathrm{ParamRandom}}
\newcommand{\setwidth}{\width}

\newcommand{\reals}{{\mathcal R}}

\newcommand{\rademacher}{{\mathcal R}}
\newcommand{\seqrademacher}{{\rademacher^{Seq}}}
\newcommand{\gaussianmean}{{\mathcal G}}
\newcommand{\myeat}[1]{}
\newcommand{\Conv}{\overline{\mathrm{Conv}}}
\newcommand{\width}{w}
\newcommand{\rangespace}{X}
\newcommand{\paramspace}{Y}

\newcommand{\EAvg}{\mathrm{Avg}}
\newcommand{\N}{\ensuremath{\mathbb{N}}}
\newcommand{\nats}{\N}

\newcommand{\myparagraph}[1]{\noindent \textbf{#1.}}

\newcommand{\idloss}{\ell_{\textrm{id}}}

\jmlrheading{23}{2025}{1-\pageref{LastPage}}{5/25}{5/25}{21-0000}{Aaron Anderson and Michael Benedikt}

\ShortHeadings{From learnable objects to learnable random objects}{Anderson and Benedikt}
\firstpageno{1}

\begin{document}

\maketitle

\begin{abstract}   
We consider the relationship between learnability of a ``base class'' of functions on a set $\rangespace$,  and learnability of a class of  statistical functions derived from the base class. For example, we refine results showing that learnability of a family $h_p: p \in \paramspace$ of functions implies learnability of the family of functions $h_\mu=\lambda p: \paramspace. E_\mu(h_p)$, where $E_\mu$ is the expectation with respect to $\mu$, and $\mu$ ranges over probability distributions on $\rangespace$. We will look at both  Probably Approximately Correct (PAC) learning, where example inputs and outputs are chosen at random, and online learning, where the examples are chosen adversarily. For agnostic learning, we establish improved bounds on the sample complexity of learning for statistical classes, stated in terms of combinatorial dimensions of the base class.  We connect these problems to techniques introduced in model theory for ``randomizing a structure''.  We also provide counterexamples for realizable learning, in both the PAC and online settings.

\end{abstract}


\section{Introduction} 
Much of classical learning theory deals with learning a function into the reals based on training examples consisting of input-output pairs. In the special case where the output space is $\{0,1\}$ we refer to a \emph{concept class}. There are many variations of the set-up. Training examples can be random -- as in ``Probably Approximately Correct'' (PAC) learning -- or they can be adversarial, as in ``online learning''. We may assume that the examples match one of the hypothesis functions (the \emph{realizable case}), or not (the \emph{agnostic case}). 

Here we will consider \emph{learning statistical objects}, where the function we are learning is itself a distribution, and we are given not individual examples about it, but statistical information.  One motivation is from database query processing, where we have queries to evaluate on a massive dataset, and we have stored some statistics about the dataset, for inputs of certain shape. For example, the dataset might be a graph with vertices having numerical identifiers, and we have computed histograms giving information about the average number of vertices connected to elements in certain intervals. In order to better estimate future queries, we may extrapolate the statistics to other unseen intervals.  
One formalization of this problem, focused specifically on learning an unknown probability distribution from statistics. was given in \citep{sigmod22}. In this setting, we have a set a points $\rangespace$,
and a collection of subsets of $\rangespace$, referred to as ranges. The random object we are trying to learn is a distribution on $\rangespace$, and we try to learn it via samples of the probabilities of ranges: that is, 
each distribution can be considered as a  function mapping a range to its probability. 
The main result of \citep{sigmod22} is that if the set to subsets is itself learnable, then the set of functions induced by distributions is learnable.

We generalize the setting of \citep{sigmod22} in several directions. We start with a   hypothesis class which can consist of either Boolean-valued or real-valued
functions on some set $\rangespace$, indexed by a parameter space $\paramspace$. We use such a ``base hypothesis class'' to form several new ``statistical hypothesis classes'', which will be real-valued functions over random objects. Two such classes are indexed by distributions $\mu$, where the corresponding functions map an input to an expectation against $\mu$.  In one class, $\mu$ represents randomization over the \emph{parameters}, and the functions will be on the input space of the original class; in the second class, $\mu$ represents randomization over \emph{range elements}.  
We show that learnability  of the base class allows us to derive learnability of the corresponding statistical classes, and establish new bounds on sample complexity of learning in terms of dimensions of the base class.

We analyze these two statistical classes using a broader result about \emph{random hypotheses classes} inspired by
work in model theory \citep{keislerrandomizing,itaykeisler,ibycontinuousrandom}. 
In the PAC learning scenario, we can apply prior work in model theory \citep{ibycontinuousrandom} to conclude that when a random family of hypothesis classes is uniformly learnable, then the expectation of this class is also learnable. We show that this implies preservation of learnability of both distribution classes. We refine these arguments in several directions:  to apply to new statistical classes, to deal with real-valued functions in the base class, and to get bounds on the number of samples needed to learn. We also examine whether the same phenomenon applies to other learning scenarios: e.g. realizable PAC learning, online learning.

\myparagraph{Contributions: preservation and sample complexity}
Our  results concern whether  various notions of learnability are preserved when moving from a base hypothesis class to the corresponding statistical class. For agnostic learning, both PAC and online, we provide  positive results on preservation, accompanied by sample complexity bounds for the statistical class, stated in terms of combinatorial dimensions of the base class. For realizable learning, we show negative results.  A high-level overview is in Table \ref{fig:summary}, which highlights the positive and negative preservation results we prove in moving from a base class to a statistical class, and the combinatorial dimension that characterizes learnability. In each of the positive results, we supply sample complexity bounds, not listed in the table. The formal definitions are in the next sections.

\begin{table*}
\addtolength{\tabcolsep}{-1pt}
\centering
{
\centering
\begin{tabular}{|c | c | c |}
\hline
Learning & Agnostic  & Realizable \\
\hline
Online &   Bounded Online FatSh \citet{rakhlin2} &  Uniform Regret or Finite Online Dim.\\
       &  Preserved (Cor. \ref{cor:preserved}) & Not Preserved for Dual Dist. Class (Prop. \ref{prop:nonclosurerealizableonline})\\
       \hline
PAC & Finite FatSh \citet{bartlettlongmore} &  Finite OIG dimension \citet{realizableonline} \\
           & Preserved (Thm \ref{thm:randomvarclass}) & Not Preserved for Distr \\
           & & or Dual Distr. Class  (Prop. \ref{prop:realizablesupnotpreserved}) \\
           \hline
    \end{tabular}
    \caption{Dimensions governing learning, and preservation moving from a base class to its statistical class} \label{fig:summary}
    }
    \end{table*}

\myparagraph{Organization} Section \ref{sec:prelims}  reviews different flavors of learnability that we study here.   Section \ref{sec:randomizedclasses} defines the notion of ``statistical class built over a base class'' that will be our central object of study.
Section \ref{sec:supervised} studies preservation of PAC learnability for statistical classes, in two variations of the learning set up: agnostic and realizable.  Section \ref{sec:online} performs the same study for online learning.
We discuss related work in Section \ref{sec:related}.
We close with conclusions  in Section \ref{sec:conc}.
We defer the proofs of a few  propositions and lemmas to the appendix.

\section{Preliminaries} \label{sec:prelims}
In our preliminaries, and elsewhere in the paper, we use \emph{fact} environments to denote results quoted from prior work.

\myparagraph{Hypothesis classes and their duals}
Our notions of learnability will be properties of a \emph{hypothesis class},  a class of functions $\hypoclass$ from some (in our case, usually infinite) set $\rangespace$, the \emph{range space of the class}, to some interval in the reals, usually $[0,1]$. 
A \emph{concept class} $\conceptclass$ is a family of subsets of $\rangespace$, which
can be considered as a hypothesis class with range $\{0, 1\}$.

Functions in a hypothesis class $\hypoclass$ are expressed
as $\hypo_p$ where $p$ ranges over the \emph{parameter space of the class} $\paramspace$. We write $\concept_p$ for a concept that is in a concept class $\conceptclass$.

A family of functions on $\rangespace$ indexed by a set $\paramspace$ can be considered as a single function from $\rangespace \times \paramspace$ to $[0,1]$.
This corresponds naturally to a function $\paramspace \times \rangespace \to [0,1]$, and thus also defines a class of functions on $\paramspace$ indexed by $\rangespace$, the \emph{dual class} of $\hypoclass$.

\myparagraph{PAC learning} We recall the standard notion of a function class being learnable \citep{kearnsschapire}
from random supervision: \emph{Probably Approximately Correct} or PAC learnable below.
Fixing $\rangespace$, let $Z=\rangespace \times [0,1]$. We call the elements of $Z$ \emph{samples}. For each
hypothesis $h \in \hypoclass$  and sample $z=(x, y)$, let $l_h(z)= (h(x)-y)^2$: this is the \emph{loss} of  using this hypothesis at sample $z$.
For a distribution $P$ over $Z$, we let $\expectedloss_P(h)$ be the expected loss of $h$ with respect to $P$.
We let $\bestexpectedloss(P)$ be the infimum of $\expectedloss_P(h)$ over every $h \in \hypoclass$.

A \emph{PAC learning procedure} is a mapping $A$ from finite sequences in $Z$ to $\hypoclass$.
For  parameters $\delta, \epsilon>0$, we say that $\hypoclass$ is \emph{$\delta,\epsilon$ agnostic PAC learnable} if there exists a
learning procedure $A$ and number $n_{\delta,\epsilon}$ such that for $n\geq n_{\delta,\epsilon}$, for every distribution $P$ over $Z$,

\[
P^n(\vec z ~ | ~ \expectedloss_P(A(\vec z)) \leq \bestexpectedloss_P + \epsilon) \geq 1 - \delta.
\] \label{eqn:paclearn}

{\noindent}Here $P^n$ denotes the $n$-fold product of $P$.

If a specific number $n_{\delta,\epsilon}$ suffices, we say that $\hypoclass$ is $\delta,\epsilon$ agnostic PAC learnable with \emph{sample complexity} $n_{\delta,\varepsilon}$. Alternately, if we just refer to bounding the $\delta,\epsilon$ sample complexity of PAC learning $\hypoclass$, we mean the smallest number $n_{\delta,\epsilon}$ that suffices.

We say $\hypoclass$ is \emph{agnostic PAC learnable} if and only if it is $\delta,\epsilon$ agnostic PAC learnable for all $\delta, \epsilon>0$.
Given a function $S(\epsilon, \delta)$ from $\epsilon, \delta \in [0,1]$ to integers, we say that $\hypoclass$ \emph{is agnostic PAC learnable with sample complexity $S$ } if for all $\epsilon, \delta \in [0,1]$, $\hypoclass$ is $\delta,\epsilon$ agnostic PAC learnable with sample complexity $S(\delta,\epsilon)$.

The qualification ``agnostic'' above refers to the fact that we do not assume that there is an unknown true hypothesis that lies in our hypothesis class. Instead, we just try to find the best approximation in our class. This contrasts with \emph{realizable PAC learning}, where
we consider a true hypothesis $h_0 \in \hypoclass$, and randomly choose only \emph{inputs}, with the outputs taken from $h_0$. Realizable PAC learning requires the learning procedure to work as above, without knowledge of $h_0 \in \hypoclass$ or the distribution, but only for samples $(x,h_0(x))$ produced by applying $h_0$ to the randomly chosen $x$. 

For us, the distinction between agnostic and realizable PAC learning will be important only in the case of real-valued functions:

\begin{fact} [See, e.g. \citet{mltheorybook}]  For a  concept class, realizable PAC learnability coincides with agnostic PAC learnability. But the sample complexity can be lower in the realizable case.
\end{fact}

\begin{fact}  \citet{realizableonline} For real-valued classes, agnostic learnability is strictly more restrictive than realizable learnability.
\end{fact}

\myparagraph{Online learning}
Another framework we consider is online learning of a hypothesis class $\hypoclass$ \citep{agnosticonline}.
Here the learner receives  examples with supervision, but  not picked randomly, but ``adversarially'': that is, arbitrarily, so the learner must consider the worst case.
A (probabilistic) online learning algorithm $A$ receives a finite sequence $s=(x_1, y_1) \ldots (x_n, y_n)$  of pairs from $\rangespace \times [0,1]$ along with an input $x$ from $\rangespace$ and returns a probability distribution over  $y \in [0,1]$. An \emph{adversary} is
an algorithm that receives a sequence of triples $s=(x_1, y_1, y'_1) \ldots (x_n, y_n, y'_n)$ and returns
a new pair $(x,y)$. Informally the first pair of each triple represents an input and real-valued output, while the last component is the value predicted by the learner.
A \emph{run} of a learning algorithm against an adversary for $T$ rounds is a sequence $s=(x_1, y_1, y'_1) \ldots (x_T, y_T, y'_T)$, where for each $i < T$, $x_{i+1}, y_{i+1}$ is chosen by running the adversary on the prefix up to $i$, and $y'_{i+1}$ is chosen
by running the learning algorithm on the concatenation of  the prefix up to $i$, projected to be a sequence of pairs, and the new adversary-generated example $(x_{i+1}, y_{i+1})$.
The \emph{loss}
 of the algorithm on such a run of length $T$ is, by default, $\Sigma_{i \leq T} | y'_i - y_{i} |$. \footnote{Here we use absolute value in online learning, but for our purposes square distance, as in PAC learning, would yield the same results.}
 The \emph{regret} for the run
 is  the difference between the loss of the algorithm and the infimum of the loss obtained by an algorithm that  uses a fixed $h \in \hypoclass$ to predict at each step. Note that if the run is highly inconsistent with all $h \in \hypoclass$, it will be much more difficult to predict; but each individual $h$ will also fail to predict well.
If we fix a strategy for the adversary, along with a probabilistic learner, we get a distribution on runs, and hence
an \emph{expected regret}. The \emph{minimax regret} is the infimum over all learning algorithms of the
supremum over all adversaries, of the expected regret.
Following \cite[Definition 2]{rakhlin2} we call a hypothesis class \emph{online learnable in the agnostic case} if there is a learning algorithm whose minimax regret against any adversary in $T$ rounds is dominated by  $T$ as $T$ goes to $\infty$.

As with PAC learning, ``agnostic'' here is contrasted to \emph{online learnability in the realizable case}. This refers to the restriction on adversaries: they must be realizable, in that each should come from applying some hypothesis $h_0 \in \hypoclass$.
We say that $\hypoclass$ is \emph{online learnable in the realizable case} if there is an algorithm that has \emph{bounded loss}, uniform in $T$, for each realizable adversary.

Although the definitions of learnability are very different, for concept classes the dividing line for learnability is the same in the realizable case 
as in the agnostic case: 
\begin{fact} \citet{agnosticonline}: 
A concept class is online learnable in the realizable case if and only if it is online learnable in the agnostic case.
\end{fact} In fact, both of these are the same as having ``bounded Littlestone dimension'' defined below.

 As with PAC learning, the dividing line for learnability diverges in the case of real-valued functions.

\myparagraph{Dimensions of real-valued families}
Each of our notions of learnability corresponds to a combinatorial dimension of the class.
For agnostic PAC learning,  the characterization involves the fat-shattering definition originating in \citep{kearnsschapire}.

\begin{definition}[Fat shattering]
For $\gamma \in (0,1]$, we say $\hypoclass$ $\gamma$ \emph{fat-shatters} a set $A \subseteq X$ if there
exists a function $s: A \rightarrow [0, 1]$ such that, for every $E \subseteq A$, there exists some $h_E\in \hypoclass$
satisfying: For every $x \in A \setminus E$, $h_E( x) \leq s( x) - \gamma$, and for every $x \in E$,
$h_E( x) \geq s(x) + \gamma$. 
The $\gamma$ fat-shattering dimension of $\hypoclass$, denoted $\pshatter_\gamma(\hypoclass)$, is the supremum of the 
cardinalities of a $\gamma$ fat-shattered subset of $X$.
\end{definition}

These dimensions give bounds on sampling:

\begin{fact}[{\citet[Thm 14]{bartlettlongmore}}] \label{fact:fatshatterandlearn}
$\pshatter_\gamma(\hypoclass)$ is finite for all $\gamma$ if and only if $\hypoclass$ is agnostic PAC learnable. In that case, there is an agnostic $\delta,\varepsilon$ PAC learning algorithm
with sample complexity
\begin{align*}
O\left({\frac{1}{\epsilon^2}} \cdot  \left( \pshatter_{\frac{\epsilon}{9}}\left(\hypoclass\right) \cdot  \log^2\left( \frac{1}{\epsilon} \right)+ \log\left(\frac{1}{\delta} \right) \right) \right).
\end{align*}

\end{fact}

The notion of dimension simplifies for a concept class.
A concept class is said to \emph{shatter} a subset $A$ of $\rangespace$ if for every $E \subseteq A$ there is
$c_A \in \conceptclass$ with $c_A$ containing $E$ and disjoint from $A \setminus E$.
The \emph{Vapnik-Chervonenkis (VC)-dimension} of $\conceptclass$ is the supremum of the cardinalities of shattered subsets.

Agnostic online learning is linked to a sequential version of the fat-shattering dimension, defined in \citep{rakhlin2}.
\begin{definition}[Sequential fat-shattering dimension and Littlestone dimension] \label{def:fatshatter}
    Let $\{-1,1\}^{<d}$ denote $\bigcup_{t = 0}^{d - 1} \{-1,1\}^t$.

    A \emph{binary tree of depth $d$ in} $X$ is a function $z : \{1,-1\}^{<d} \to X$. Given such a binary tree and $t < d$, let $z_t$ be the restriction of $z$ to inputs in $\{1,-1\}^t$. If $\branch \in \{1,-1\}^d$ is a branch, we write $z_t(\branch)$ to denote $z_t$ applied to the restriction of $\branch$ to the first $t$ entries.

    Branches $B \in \{-1,1\}^d$ can also be viewed as functions $B: \{0,\dots, d-1\} \to \{-1,1\}$, with $B(t)$ the $t^{th}$ entry of $B$.

    For $\gamma \in (0,1]$ say $\hypoclass$ $\gamma$ \emph{fat-shatters a binary tree $z$ in $X$}, where $z$ is of depth $d$,   if there
exists a binary tree $s$ of depth $d$  in $\mathbb{R}$ and a labelling of each $\branch \in \{1,-1\}^d$ with some $h_\branch \in \hypoclass$
satisfying: For every $0 \leq t < d$, if $\branch(t) = -1$, then $h_\branch( z_t(\branch) ) \leq s_t(\branch) - \frac{\gamma}{2}$ and, if $\branch(t) = 1$, then
$h_\branch(z_t(\branch)) \geq s_t(\branch) + \frac{\gamma}{2}$. 
The $\gamma$ \emph{sequential fat-shattering dimension} of $\hypoclass$, denoted $\seqfatshatter_\gamma(\hypoclass)$, is the supremum of the 
depths of $\gamma$ fat-shattered binary trees in $X$.

In the discrete case, a concept class shatters a binary tree $T: \{-1,1\}^{<n} \to X$ when for every branch $\branch \in \{-1,1\}^n$ of the tree, there is $c_\branch \in\conceptclass$ such that for all $k < n$, $T(\branch|_{\{-1,1\}^k})$ is in $c_A$ if and only if $\branch(k) = 1$. The \emph{Littlestone dimension} of $\conceptclass$ is the supremum of the depths of shattered binary trees.
\end{definition}

Littlestone dimension  characterizes online learnability (realizable or agnostic), in the case of concept classes, analogously to the way VC dimension characterizes PAC learnability (realizable or agnostic):
\begin{fact}[{\citet[Theorem 12.1]{alon_adversarial}}]\label{fact:littlestone_online}
    If $\hypoclass$ is a $\{0,1\}$-valued concept class with Littlestone dimension at most $d$, then the minimax regret of a $T$-round online learner is bounded by
    $$O(\sqrt{dT}).$$ Conversely, if the dimension is infinite, then the minimax regret is infinite.
\end{fact}

In the real-valued case, the bound is more complicated, but finiteness of all sequential fat-shattering dimensions is equivalent to agnostic online learnability:
\begin{fact}[{\citet[Part of Proposition 9]{rakhlin2}}]\label{fact:sfs_online}
    If $\hypoclass$ is a hypothesis class taking values in $[0,1]$, then the minimax regret of a $T$-round online learner is bounded below by
    $$\frac{1}{4 \cdot \sqrt{2}}\sup_\gamma \min\left(\sqrt{\seqfatshatter_\gamma(\hypoclass) \cdot T},T\right)$$
    and above by
    \begin{small}
    $$\inf_\gamma \left(4 \cdot \gamma \cdot T + 12 \cdot \sqrt{T} \cdot \int_\gamma^1 \sqrt{\seqfatshatter_\beta(\hypoclass)\log\left(\frac{2 \cdot e \cdot T}{\beta}\right)}\,\mathrm{d}\beta\right).$$
    \end{small}

    In particular, the minimax regret is sublinear if and only if $\seqfatshatter_\gamma(\hypoclass)$ is finite
    for every $\gamma$.
\end{fact}

\myparagraph{Duality and agnostic learnability}
It is well-known that agnostic PAC learnability is closed under moving to the dual:

\begin{fact} \citet{dual} \label{fact:pacdual} A function class is agnostic PAC learnable exactly when its dual class is.
\end{fact}

Using the combinatorial characterizations, one can show that same for online learning  (see Appendix \ref{app:dualityagnosticonline}):

\begin{proposition} \label{prop:onlinedual} A function class is agnostic online learnable exactly when its dual class is.
\end{proposition}

\section{Hypothesis classes from statistical objects} \label{sec:randomizedclasses}

We now turn to the basic object of study in our paper:

\begin{definition}[Distribution class and Dual Distribution class] Given a concept class $\conceptclass$ on $\rangespace$, parameterized by $\paramspace$ 
the \emph{distribution function class of $\conceptclass$}, denoted $\measureclass{\conceptclass}$, is a real-valued hypothesis
class on the same range space $\rangespace$, consisting of the hypotheses $h_\mu$ indexed by the
distributions $\mu$, with $\sigma$-algebra $\Sigma_\mu$, on $\paramspace$ such that for each $x$, the set of $p \in \paramspace$ such that $c_p(x)=1$ is measurable. 

The hypothesis $h_\mu$ is defined as the function mapping $x$ to the probability, with respect to $\mu$, that $c_p(x)=1$. 
Note that, WLOG, we can assume that the $\Sigma_\mu$ is the $\Sigma$-algebra generated by $\{\{p \in \paramspace : c_p(x) = 1\} : x \in \rangespace\}$, since  $h_\mu$ is completely determined by the restriction of $\mu$ to this $\Sigma$-algebra.

The \emph{dual distribution function class of $\conceptclass$}, denoted $\dualmeasureclass{\conceptclass}$, 
is defined as above, but starting
with the dual class of $\conceptclass$. That is, 
a hypothesis is parameterized by a distribution $\mu$ on $\rangespace$,
such that each $C_p \in \conceptclass$ is measurable.
The function $h_\mu$ maps $p \in \paramspace$ to  to $\mu(C_p)$.

We extend these definitions to work on top of a real-valued function class $\hypoclass$, replacing probability under $\mu$ with expectation. For the distribution
function class, the functions are indexed by distributions $\mu$ on $\paramspace$ again, but now we restrict to distributions $\mu$ such
that each $h \in \hypoclass$ is a $\mu$ measurable function,  and $h_\mu$ maps $h \in \hypoclass$ to $E_\mu(h)$.
As with the concept class, we can WLOG assume that $\Sigma_\mu$ is the $\sigma$-algebra generated by the sets $\{y : h_y(x) \in I\}$ where $I$ is an interval, $x \in \rangespace$.
\end{definition}

\begin{myexmp} \label{ex:rectangles} Let $\hypoclass$ be the concept class of rectangles over the reals. The range space is thus the collection of points in the plane -- pairs of reals --  while the parameter space consists of $4$-tuples of reals, representing the lower left corner and upper right corners of the rectangle. 

The dual distribution class of $\hypoclass$ will be parameterized by ``random elements of the plane'': functions induced by distributions $\mu$ over the plane. Each such $\mu$ induces a function $h_\nu$ on rectangles, mapping each rectangle to its $\nu$-probability. We can equivalently consider $h_\mu$ as a real-valued function on $4$-tuples.  In learning such a function $h_\mu$, we will be given supervision in the form of a sequence of pairs, each pair consisting of
a rectangle (or $4$-tuple) and the $\mu$-probability that a point is in the rectangle.

In contrast, the distribution class of $\hypoclass$ is parameterized by ``random rectangles'': functions induced by distributions $\nu$ over $4$-tuples.
Given such a $\nu$, we have a function $h_\nu$ that maps 
a real pair $\vec x$ to the $\nu$ probability that $\vec x$ is in rectangle $h_\nu$.
In learning such a function, our supervision will consist of
a point in the real plane and the probability that the point is in random rectangle $\nu$.
\end{myexmp}

\begin{myexmp} \label{ex:deffunction} Consider the hypothesis class $\hypoclass$ consisting of  rational  functions $\frac{P(x)}{Q(x)}$, where $P(x)$ and $Q(x)$ are real polynomials restricted to an interval where $Q$ has no zeros (with the function defined to be zero off this interval), the degrees of $P$ and $Q$ are at most $5$, and the value of the quotient function on the interval is in $[0,1]$. This is a class with parameters $\vec p$  representing the coefficients of $P$ and $Q$.
We can easily show (see the appendix for a broader argument)  that this class of real-valued functions indexed by $\vec p$ is agnostic PAC learnable.

The dual distribution class is parameterized by ``random arguments'' to these functions: with each distribution inducing a function on parameters given by integrating against the distribution.
The distribution class will be parameterized by ``random coefficients'': each distribution will induce a function on arguments $x$ obtained by integrating the random function against the distribution on parameters.
\end{myexmp}

Only the dual distribution class has been studied in the past literature, and exclusively for the case of a concept class. The following result is our starting point: 

\begin{fact}\citet{sigmod22} \label{fact:sigmod22}
Suppose a concept class  $\conceptclass$ is (Agnostic or Realizable) PAC learnable, and the VC-dimension of the class is $\lambda$. Then  the dual distribution function class of $\conceptclass$ is Agnostic PAC learnable with sample complexity $\tilde{O}(\frac{1} {\epsilon^{\lambda+1}})$ where $\tilde{O}$ indicates that we drop terms that are polylogarithmic in $\frac{1}{\epsilon}$,
$\frac{1}{\delta}$ for constant $\lambda$.
\end{fact}

The motivation in \cite{sigmod22} was from databases: we want to learn the selectivity of a table in a databases. Given some examples of the density of the table in some intervals, we want to predict the density in a new interval. \cite{sigmod22} deals with computational complexity as well as sample complexity, but here we will only be interested in sample complexity bounds.

\subsection{Expectation of a Measurable Family of Hypothesis Classes}
We will show that we can embed the distribution class and the dual distribution class in a more general construction for hypothesis classes, where parameters and range elements are treated symmetrically. This construction is fundamentally connected to the \emph{randomization} construction for hypothesis classes coming from logical
formulas, originating with \citep{keislerrandomizing}, refined later in \citep{itaykeisler, ibycontinuousrandom}.

We will start with a construction that works on something much more general than a hypothesis class.

\begin{definition}[Measurable family of hypothesis classes]
    Assume that $(\Omega, \Sigma, \mu)$ is a probability space, and $\classfamily = (\hypoclass_\omega : \omega \in \Omega)$ is a family of hypothesis classes $\hypoclass_\omega : \rangespace \times \paramspace \to [0,1]$ such that for each $x \in \rangespace, y\in \paramspace$, the function $\omega \mapsto \hypoclass_\omega(x,y)$ is measurable.
    We call such a family a \emph{measurable family of hypothesis classes on $\rangespace$ parameterized by $\paramspace$}.
\end{definition}

Thus a measurable family can be thought of as a \emph{randomized hypothesis class}, where for each $x \in \rangespace, y \in \paramspace$ we randomly choose a real-valued output.

\begin{definition}[Expectation class of a measurable family of hypothesis classes]
 Let $\classfamily = (\hypoclass_\omega : \omega \in \Omega)$ be a measurable family of hypothesis classes. Then the \emph{expectation class} of $\classfamily$ is the function $\expclass\classfamily : \rangespace \times \paramspace \to [0,1]$ defined by 
    $$\expclass\classfamily(x,y) = \mathbb{E}_\mu[\hypoclass_\omega(x,y)].$$
    This is a single hypothesis class, with range space $\rangespace$ and parameter space $\paramspace$.
\end{definition}

\subsection{Embedding Distribution Classes in Expectation Classes} \label{subsec:embedding}

\begin{definition}[$X$-valued random variables]
    If $(\Omega,\Sigma,\mu)$ is a probability space and $(X,\Sigma_X)$ is a measurable space, then let $\randvar_X$ be the set of $X$-valued random variables over $\Omega$: that is, measurable functions from $\Omega \to X$.
\end{definition}

\begin{definition}[Compatibility with a hypothesis class;
Randomized version of a hypothesis class]
Let $\hypoclass$ be a hypothesis class on range space $\rangespace$ parameterized by $\paramspace$, and let
$\Omega^*=(\Omega, \Sigma, \mu)$ be a probability space.
A $\paramspace$-valued random variable $Y'$ is said to be \emph{compatible with $\hypoclass$} if for each $x \in \rangespace$
the function mapping $\omega \in \Omega$ to $h_{Y'(\omega)}(x)$ is measurable in the usual sense.

We write $\compatparamrv(\hypoclass)$ for the set of  $\paramspace$-valued random variables that are compatible with a given $\hypoclass$, omitting the dependence on $\Omega^*$ for brevity.

    For a given $Y'\in \compatparamrv(\hypoclass)$ we refer to the real-valued function on $\rangespace$ above as the \emph{$Y'$-randomized version of $\hypoclass$}.
    \end{definition}
Thus each compatible $\paramspace$-valued random variables thus give an alternative notion of ``random parameter''.

\begin{definition}[Parameter randomized version of a hypothesis class]\label{def:paramrandomizedfamily}
Given a hypothesis class $\hypoclass$ on range space $\rangespace$ parameterized by $\paramspace$ and probability space
$(\Omega, \Sigma, \mu)$, we define a measurable family, which we will call the \emph{parameter randomized family of $\hypoclass$}, denoted $\paramrandomized(\hypoclass) = ((\paramrandomized(\hypoclass))_\omega : \omega \in \Omega)$.
Given $\omega \in \Omega$, let $(\paramrandomized(\hypoclass))_\omega$ be the hypothesis class on $\rangespace$ consisting of the functions:
For each $Y' \in  \compatparamrv(\hypoclass)$, the function $\lambda x \in \rangespace. \hypo_{Y'(\omega)}(x)$. 
\end{definition}

Finally, we can define the expectation class that we want:
\begin{definition}[Parameter randomized expectation class] \label{def:paramrandomizedexpectationclass}
We consider the hypothesis class denoted $\hypoclass(\compatparamrv(\hypoclass))$: it is on the same range space $\rangespace$, but parameterized by $\compatparamrv(\hypoclass)$. Given $Y' \in \compatparamrv(\hypoclass)$ $\hypo_{Y'}$ maps $x \in \rangespace$ to $E_\mu$ of the $Y'$-randomized version of $\hypoclass$.
That is, it is the expectation class of the parameter randomized family of $\hypoclass$.
\end{definition}

The parameter randomized expectation class of a hypothesis class $\hypoclass$ is very similar to the distribution class of $\hypoclass$. The  difference is that the former considers $\paramspace$-valued functions from a fixed probability space, while the latter deals with the induced distributions on $\paramspace$, ignoring the underlying sample space. We now show that by choosing the probability space carefully, we can close the gap.

Let $\nu$ be a distribution on $\paramspace$, where the underlying $\Sigma$-algebra is generated by $S_{x,I}=\{p \in \paramspace | \hypo_p(x) \in I\}$ for $I$ an interval and $x \in \rangespace$. We say $\nu$ is \emph{induced} by a function $f$ from
$\Omega$ to $\paramspace$ (with respect to $(\Omega,\Sigma, \mu)$) if $\nu(S_{x,I})=\mu(\{\omega| f(\omega) \in S_{x,I}\}$.

The following simple result in measure theory states that by choosing $(\Omega,\Sigma, \mu)$ appropriately, all the distributions over
$\paramspace$ are induced by functions from $(\Omega,\Sigma, \mu)$:

\begin{proposition} \label{prop:induced} For any $\paramspace$ we can find $(\Omega,\Sigma, \mu)$ such that for every distribution
$\nu$ over the $\Sigma$-algebra on $\paramspace$ generated by $S_{x,I}$, there is a function $f$ such that $\nu$ is induced by
$f$ with respect to $(\Omega,\Sigma, \nu)$.
\end{proposition}
Although the result is well-known (see, e.g. \cite{fremlin}), we sketch the proof:
\begin{proof} 
We choose $\Omega^*=(\Omega, \Sigma, \mu)$ to be a product of all measure spaces on $\paramspace$. Then every measure is induced
as a projection on one component.
\end{proof}

From the proposition we see: 
\begin{proposition} \label{prop:embeddingbottomline} By choosing $\Omega^*$ as in Proposition \ref{prop:induced}, we have $\hypoclass(\compatparamrv(\hypoclass))$ contains all the functions in the distribution class of $\hypoclass$. Thus if the former is learnable, in any of the variants we have discussed (agnostic PAC, agnostic online etc.),  then so is the latter, with the same sample complexity bounds.
\end{proposition}

We can analogously talk about $\rangespace$-valued random variable $X'$ being compatible with $\hypoclass$, and define
the class  $\hypoclass(\compatrangerv(\hypoclass))$ of hypothesis parameterized by such random variables: these subsume
the dual distribution class.

We will use these embeddings below to transfer results about learnability of the expectation class of a measurable family of hypothesis classes to results about the distribution and dual distribution classes.
We emphasize that the notion of measurable family is much more general than the distribution class and dual distribution class.
In the distribution class, only parameters are randomized, while range elements are deterministic, while in the dual distribution class it is the reverse. The measurable family  allows us to model situations where the range elements and parameters are both randomized, in a correlated way.

\begin{myexmp} \label{ex:arith}
Consider a concept class $\conceptclass$ containing characteristic
functions $h(x)_{y_1, y_2}$ where both the range elements $x$ and the parameters $y_1, y_2$
are integers. The concept
contains any $x$ that an even integer between $y_1$ and $y_2$.
Thus as we vary the parameters $y_1, y_2$ we get the set of intervals intersected with the even numbers.

Fix a probability space $\Omega^*=(\Omega, \Sigma, \mu_0)$. The parameterized randomized family of $\hypoclass$ is a measuable family of hypothesis class parameterized by a ``random pair of integers'': a  measurable function taking
 $\omega \in \Omega$ to a pair $y_1, y_2$. The measurable family contains a class for each $\omega \in \Omega$ parameterized
 by such functions $P$, where for each $\omega$ and $P$, the function maps $x$ to  $C(x)_{P(\omega)}$.
The expectation class of this measurable family is the distribution class of the concept class $\conceptclass$. 

If we reverse the role of parameters and range values, we get a measurable family of concept classes indexed by random integers $X$, where given $\omega$ and $X$ we map $(y_1, y_2)$ to $C(X(\omega))_{y_1, y_2}$. The expectation class of this measurable
family is the dual distribution class of $\conceptclass$

We could also consider measurable families parameterized by functions from $\omega$ to triples $(x,y_1, y_2)$, which represent a correlated pair $X,P$ of random range elements and random parameters.
\end{myexmp}



While the notation used in \cite{ibycontinuousrandom} is different, it is proven there that a uniform bound on fat-shattering dimensions of the hypothesis classes in a measurable family implies finite fat-shattering dimension, and thus PAC learnability, of the expectation class.

\begin{definition} For $\gamma>0$, say that a measurable family of hypothesis classes  $\classfamily$ has \emph{uniformly bounded $\gamma$ fat-shattering dimension} it there is some $d$ such that each class $\hypoclass \in \classfamily$ has $\gamma$ fat-shattering dimension at most $d$.
\end{definition}
\begin{fact}[{\citet[Corollaries 4.2, 4.3]{ibycontinuousrandom}}] \label{fact:itaymeasurablepac}
    If $\classfamily$ is a measurable family of hypothesis classes, and for every $\gamma > 0$, the family has uniformly bounded $\gamma$ fat-shattering dimension, then
    $\expclass\classfamily$ has finite fat-shattering dimension.

\end{fact}

Now let us return to the distribution class, and recall that it can be embedded in the parameter randomized family of $\hypoclass$, which associates the function $\lambda x \in \rangespace. \hypo_{Y'(\omega)}(x)$.
If $\hypoclass$ has finite $\gamma$ fat-shattering dimension $d$, then for any $\omega$, the corresponding class of functions above also has fat-shattering dimension $d$: thus the measurable family has uniformly bounded $\gamma$ fat-shattering dimension. Similarly for the dual distribution class.
Thus we get the following corollary of Fact \ref{fact:itaymeasurablepac}:
\begin{corollary} If $\hypoclass$ is agnostic PAC learnable, so are the distribution class and dual distribution classes.
\end{corollary}

In Section \ref{sec:supervised}, we will adapt the proof from \cite{ibycontinuousrandom} to provide concrete sample complexity bounds on learnability of the expectation class of a measurable family, which will provide improved bounds for the distribution and dual distribution class.

\section{PAC learning of statistical classes} \label{sec:supervised}

This section will be devoted to a more fine-grained investigation of how PAC learnability for these statistical classes follow from PAC learnability of the base class.
Our first main result is a new proof that agnostic PAC learnability is preserved by moving to any of these classes, along with a new bound on sample complexity in terms of dimensions of the base class:

\begin{remark} A warning that \emph{throughout this section and the next, we ignore several measurability issues that arise when uncountable hypothesis classes are considered -- e.g. measurability of sets that involve intersecting over uncountabiy many objects.} These subtleties do not arise when the parameter set and range sets are countable, and all the results in the next sections are true without qualification in this case. Extensions to the uncountable case require further sanity conditions on the hypothesis classes. We defer a discussion of this to Appendix \ref{app:measurability}.
\end{remark}

\begin{restatable}{theorem}{supervisedrandomvarclass} \label{thm:randomvarclass}
If $\classfamily$ is a measurable family of hypothesis classes such that each class $\hypoclass \in \classfamily$ has $\frac{\epsilon}{50}$ fat-shattering dimension at most $d$, one can perform agnostic PAC learning on the expectation class $\expclass\classfamily$ with sample complexity:
$$O\left(\frac{d}{\epsilon^4} \cdot \log^2\frac{d}{\epsilon} + \frac{1}{\epsilon^2} \cdot \log \frac{1}{\delta}\right).$$

If $\classfamily$ is a measurable family of $\{0,1\}$-valued concept classes
with VC-dimension at most $d$,
one can perform agnostic PAC learning on the expectation class $\expclass\classfamily$ with sample complexity:

$$O\left(\frac{1}{\epsilon^2}\left(d \log\frac{d}{\epsilon} + \log\frac{1}{\delta}\right)\right).$$
Thus, via Proposition \ref{prop:embeddingbottomline},  we get the same bounds for the distribution class and (moving to the dimension of the dual class) for the dual distribution class.
\end{restatable}

\begin{myexmp} \label{ex:followup}
Recall Example  \ref{ex:deffunction}, where we consider a family $\hypoclass$ of rational functions definable by real-coefficients $\vec p$.
The distribution function class of $\hypoclass$ is parameterized by random $\vec p$. It is agnostic PAC learnable with sample complexity as in Theorem \ref{thm:randomvarclass}. Thus given supervision based on the expectations for various $x_i$, we can learn the hypothesis in the distribution class that has the best expected fit in terms of sum of differences.
%
\end{myexmp}

Thus the next few subsections will be devoted to the proof of Theorem \ref{thm:randomvarclass}. which will complete our sample complexity analysis of agnostic PAC learning for statistical classes.

\subsection{Combinatorial and statistical tools}
Our arguments for PAC learnability will go through the following dimension:

\begin{definition}[Glivenko-Cantelli dimension]
The Glivenko-Cantelli dimension of a hypothesis class, denoted $\glivcant_{\hypoclass}(\epsilon, \delta)$ is parameterized by $\delta, \epsilon>0$:

\begin{align*}
\glivcant_{\hypoclass}(\epsilon, \delta)  = 
\min ~
\{n : \forall m\geq n, \forall D \mbox{ Distribution on } X ~ \\
D^m \left\{(x_1,..., x_m) ~ |~ \exists \hypo \in \hypoclass ,
\left| \frac{1}{m} \cdot (\Sigma^m_{i=1} \hypo(x_i) ) - \int \hypo(u) dD(u)\right| > \epsilon \right\} 
\leq \delta \}
\end{align*}
\end{definition}

Recall that the law of large numbers implies that if we fix any
bounded measurable function $f$ into the reals, and are given a $\delta$ and $\epsilon$, then  we can find an $n$ so that, for any distribution
$D$, for all but $\delta$ of
the $n$-samples from the distribution, the sample mean of $f$ is within $\epsilon$ of the the mean of $f$.
Most proofs that a class is agnostic PAC learnable go through showing that for each $\epsilon,\delta > 0$, the dimension $\glivcant_{\hypoclass}(\epsilon,\delta)$ is finite. From the Glivenko-Cantelli dimension for $\epsilon$ and $\delta$, we can easily obtain bounds on the number of samples needed to learn to given tolerances $\delta$ and $\epsilon$.

Glivenko-Cantelli bounds are used to derive learnability bounds in \citep{alon} and \citep[Theorem 14]{bartlettlongmore}. The proof of \citep[Theorem 19.1]{bartlettbook} gives the following version of the connection between these concepts:
\begin{fact}[\citet{bartlettbook}]\label{fact:gc_sample}
    The sample size needed to learn $\hypoclass$ with error at most $\epsilon$ and error probability at most $\delta$ is at most $GC_\hypoclass\left(\frac{\epsilon}{2},\frac{\delta}{2}\right)$.
\end{fact}
Because the bounds we will derive for $GC_\hypoclass(\varepsilon,\delta)$ will always be polynomial in $\varepsilon$ and $\delta$, the factors of $2$ in this fact will only change the bound by a constant multiple, which will only change the constant within asymptotic notation.

We will follow the methods of \citep{ibycontinuousrandom}, which relate upper and lower bound on the combinatorial fat-shattering dimensions to upper and lower bounds on the
\emph{Rademacher mean width}. We will be able to then move to bounds on the Rademacher width of the expectation class, and from there to bound on GC-dimension of the expection class.

\begin{definition} \label{def:heightbody}[Width of a set in a direction]
  Let $A \subseteq \reals^n$ be bounded. Given $\vec b \in \reals^n$, we define $\setwidth(A,\vec b)$, the \emph{width} of $A$ in the particular direction $\vec b$, to be
    $\setwidth(A,\vec{b}) = \sup_{\vec a \in A}\vec a \cdot \vec b.$
    \end{definition}
    \begin{lemma}\label{lem:height_measurable} For any bounded $A \subseteq \reals^n$, the width function $\vec b \mapsto \width(A,\vec b)$ is Borel measurable.
\end{lemma}
\begin{proof}
    If $A$ is countable, then $\width(A,\vec b)$ is the supremum of a countable family of continuous functions of $\vec b$, and is thus Borel measurable.

    If $D \subseteq A$, then for every $\vec b$, $\setwidth(D,\vec b) \leq \setwidth(A,\vec b)$. If $D$ is dense in $A$, then for every $\vec a \in A$, $\vec b \in \reals^n$, and $\epsilon > 0$, there is some $\vec d \in D$ such that
    $|\vec b \cdot (\vec a - \vec d)| \leq \epsilon$, so we can conclude that
    $\setwidth(D,\vec b) \geq \setwidth(A,\vec b) - \epsilon$, so in fact, $\setwidth(D,\vec b) = \setwidth(A,\vec b)$.

    Every subspace of $\reals^n$ is separable, so for every $A$ there is a countable dense $D$, and $\vec b \mapsto \width(D,\vec b)$, and thus $\vec b \mapsto \width(A,\vec b)$, is measurable.
\end{proof}
\begin{definition}\label{def:meanwidth}[Mean width]
    Let $\beta$ be a Borel probability measure on $\reals^n$. Define the \emph{mean width of $A$ w.r.t. $\beta$}, $w(A,\beta)$, as
    $\mathbb{E}_{\beta}\left[\setwidth(A,\vec b)\right]$, where $\vec{b}$ is a random variable with distribution $\beta$.
 By the measurability result mentioned above, this is always defined.

If $\beta$ is the distribution that samples uniformly from $\{+1,-1\}^n$, we define the \emph{Rademacher mean width}, denoted $w_\rademacher(A)$, as $w(A,\beta)$.

For any function $\hypo : \rangespace \to [0,1]$ and $\bar x = (x_1,\dots, x_n) \in \rangespace^n$, we let $h(\bar x)$ denote
the vector $(\hypo(x_1),\dots,\hypo(x_n))$.
In particular, if $\hypoclass$ is a hypothesis class on $\rangespace$ parameterized by $\paramspace$, $\bar x = (x_1,\dots, x_n) \in \rangespace^n$, and $y \in \paramspace$, then $\hypoclass(\bar x,y)=(\hypoclass(x_1,y),\dots,\hypoclass(x_n,y))$.

For any hypothesis class $\hypoclass$ on $\rangespace$ parametrized by $\paramspace$ and $\bar x = (x_1,\dots, x_n) \in \rangespace^n$, let $\hypoclass(\bar x,\paramspace)$ be the set of vectors $$\{\hypo(\bar x): \hypo \in \hypoclass\} = \{\hypoclass(\bar x,y): y \in \paramspace\} \subseteq [0,1]^n.$$

Then we extend the definition of Rademacher mean width to be a function of an integer $n$, given
the class $\hypoclass$ on $\rangespace$ parametrized by $\paramspace$; $\rademacher_\hypoclass(n) = \sup_{\bar x \in X^n}w_\rademacher(\hypoclass(\bar x,\paramspace))$.
We will refer to this function as the Rademacher mean width of the class $\hypoclass$, but this only differs from the ``Rademacher complexity'' of $\hypoclass$ in other literature by a factor of $n$.
\end{definition}

The reason for looking at Rademacher mean width will be that it behaves well under averaging with respect to an arbitrary measure: see Theorem \ref{thm:exp_mean_width} to follow.

From a bound on the Rademacher width of a function class,
we can infer a bound on the Rademacher width of its expectation.
Using this and some relationships between Rademacher width bounds
and fat-shattering, we are able to bootstrap from uniform bounds on a measurable family to bounds on the expectation class, proving Theorem \ref{thm:randomvarclass}.

\subsection{Bounding the mean width for a derived class in terms of a base class}
In this subsection, we will establish connections between combinatorial dimensions of each class of sets or functions in a measurable family and the dimensions of the expectation class of the family. 
Following the approach in \citep{ibycontinuousrandom}, we will first establish this connection for notions of mean width. 

We will show that Rademacher mean width does not increase under averaging.
With that in hand, 
if we are able to bound learnability of each class in a family $\classfamily$ through mean width, the same bound will apply to $E\classfamily$.
This strategy stems from \cite[Theorem 4.1] {ibycontinuousrandom}, where it was applied to Gaussian mean width.

We can now make a note about how width of the set of vectors induced by a measurable family behaves under expectation.
\begin{lemma}\label{lem:exp_height}
    Let $(\Omega,\Sigma,\mu)$ be a probability space, and let $\classfamily = (\hypoclass_\omega : \omega \in \Omega)$ be a  measurable family of hypothesis classes on $\rangespace$. Fix $\bar x \in \rangespace^n$ and $\vec b \in \reals^n$. Then 
    $$\setwidth(\expclass\classfamily(\bar x,\paramspace),\vec b) \leq \mathbb{E}_\mu[\setwidth(\hypoclass_\omega(\bar x,\paramspace),\vec b)].$$
\end{lemma}
Here recall from Definition \ref{def:meanwidth} that $\hypoclass_\omega(\bar x,\paramspace)$ is the image of
    the set of vectors for hypothesis class $\hypoclass_\omega$ at $\bar x$, as $y$ ranges over $\paramspace$.
\begin{proof}
    The supremum of the expectations of a family of functions is at most the expectation of their suprema, so we have
    \begin{align*}
        \setwidth(\expclass\classfamily(\bar x,\paramspace),\vec b)
        &= \sup_{y \in Y}\vec b \cdot (\expclass\classfamily(\bar x,y))\\ 
    &= \sup_{y \in Y}\mathbb{E}_\mu[\vec b \cdot (\hypoclass_\omega(\bar x,y))]\\
    &\leq \mathbb{E}_\mu\left[\sup_{y \in Y}\vec b \cdot  \hypoclass_\omega(\bar x,y)\right]\\
    &= \mathbb{E}_\mu[\setwidth(\hypoclass_\omega(\bar x,\paramspace),\vec b)].\\
    \end{align*}
which proves the lemma.
    \footnote{This result is implicit in the proof of \cite[Theorem 4.1]{ibycontinuousrandom}.
    For a fixed $\bar x$, it is stated there that
    $\expclass\classfamily(\bar x,\paramspace) \subseteq \mathbb{E}_\mu[\Conv(\hypoclass_\omega(\bar x, Y))]$,
    where the latter expectation is an expectation \emph{of convex compact sets}.
    This amounts to saying that for every $\vec b \in \reals^n$,
    $$\setwidth(\expclass\classfamily(\bar x,\paramspace),\vec b) \leq \mathbb{E}_\mu[\setwidth(\hypoclass_\omega(\bar x,\paramspace),\vec b)].$$}
\end{proof}


And we can now extend that statement about width of sets to a statement about mean width.
\begin{lemma}\label{lem:exp_mean_width}
    Let $(\Omega,\Sigma,\mu)$ be a probability space, and let $\classfamily = (\hypoclass_\omega : \omega \in \Omega)$ be a  measurable family of hypothesis classes on $\rangespace$.
    
    Fix $n$, $\bar x \in \rangespace^n$, and a Borel probability measure $\beta$ on $\reals^n$. Then 
    $$\width(\expclass\classfamily(\bar x,\paramspace),\beta) \leq 
    \mathbb{E}_\mu[\width(\hypoclass_\omega(\bar x,\paramspace),\beta)].$$
\end{lemma}

\begin{proof}
    To prove this, we only need to unfold the definition of $\width(A,\beta)$ and apply Lemma \ref{lem:exp_height} and Fubini's Theorem. 
    \begin{align*}
        w(\expclass\classfamily(\bar x,\paramspace),\beta)
        &= \mathbb{E}_\beta\left[\width(\expclass\classfamily(\bar x,\paramspace),\vec b)\right]\\
        &\leq \mathbb{E}_\beta \mathbb{E}_\mu[\width(\hypoclass_\omega(\bar x,\paramspace),\vec b)]\\
        &= \mathbb{E}_\mu \mathbb{E}_\beta[\width(\hypoclass_\omega(\bar x,\paramspace),\vec b)]\\
        &= \mathbb{E}_\mu\left[\width(\hypoclass_\omega(\bar x,\paramspace),\beta)\right]
    \end{align*}
\end{proof}

We are now ready to bound the Rademacher mean width of an expectation using the Rademacher mean width of the underlying class:

\begin{theorem}[Pushing a Mean Width Bound through an Expectation] \label{thm:exp_mean_width}
    Let $(\Omega,\Sigma,\mu)$ be a probability space, and let $\classfamily = (\hypoclass_\omega : \omega \in \Omega)$ be a measurable family of hypothesis classes on $\rangespace$.
    
    Then
    $$\rademacher_{\expclass\classfamily}(n) \leq \sup_\omega \rademacher_{\hypoclass_\omega}(n).$$
\end{theorem}

\begin{proof}
    For each $n$, where $\beta$ is uniformly distributed on $\{-1,1\}^n$, it suffices to show that
    $$\sup_{\bar x \in \rangespace^n}\setwidth(\expclass\classfamily(\bar x,\paramspace),\beta)
    \leq \sup_\omega \sup_{\bar x \in \rangespace^n} \setwidth(\hypoclass_\omega(\bar x,\paramspace),\beta),$$
    which, as the suprema commute, amounts to showing that for each $\bar x \in \rangespace^n$,
    $$\setwidth(\expclass\classfamily(\bar x,\paramspace),\beta)
    \leq \sup_\omega \setwidth(\hypoclass_\omega(\bar x,\paramspace),\beta),$$
    which follows from Lemma \ref{lem:exp_mean_width}
    as $$\mathbb{E}_\mu[\setwidth(\hypoclass_\omega(\bar x,\paramspace),\beta)] \leq \sup_\omega w(\hypoclass_\omega(\bar x,\paramspace),\beta).$$
\end{proof}

\subsection{Glivenko-Cantelli Bounds through Mean Width}
Above we have seen how to estimate how moving to an expectation impacts means width.  We will now look at the impact on Glivenko-Cantelli dimension. We can bound the Glivenko-Cantelli dimension with Rademacher mean width. The following is a restatement of \cite[Theorem 4.10]{wainwright} in terms of $GC$-dimension, using the fact that for any probability measure $\mu$ on $\rangespace$ and any hypothesis class $\hypoclass$ on $\rangespace$ parameterized by $\paramspace$, the
\emph{Rademacher complexity}
$\frac{1}{n}\mathbb{E}_{\mu^n}[w_\rademacher(\hypoclass(\bar x,\paramspace))]$ is at most $\frac{1}{n}\rademacher_\hypoclass(n)$.

\begin{fact} \label{fact:rademacher}
    Let $\hypoclass$ be a hypothesis class on $\rangespace$ parameterized by $\paramspace$. For any $\delta > 0$ and $n$, then
    $$GC_\hypoclass\left(\frac{2 \cdot \rademacher_\hypoclass(n)}{n} + \delta,\exp\left(-\frac{n\delta^2}{2}\right)\right) \leq n.$$
\end{fact}

We can rephrase this fact in a form that makes it easier to calculate the Glivenko-Cantelli dimension.
\begin{lemma}[From Rademacher Width of a Base Class to GC of the Expectation class]\label{lem:gc}
    
    Let $(\Omega,\Sigma,\mu)$ be a probability space, and let $\classfamily = (\hypoclass_\omega : \omega \in \Omega)$ be a measurable family of hypothesis classes on $\rangespace$.
    
    For any $\epsilon, \delta > 0$, if $N$ is such that for all $n \geq N$,
    $\frac{\rademacher_{\hypoclass_\omega}(n)}{n} \leq \frac{\epsilon}{4}$ for each $\omega \in \Omega$, then
    $$GC_{\expclass\classfamily}(\epsilon,\delta) \leq N + \frac{8}{\epsilon^2} \log \frac{1}{\delta}.$$
\end{lemma}

Roughly speaking, the lemma says that, when fixing $\epsilon$, if we can find a linear bound on the Rademacher mean width, then we can bound the $\epsilon, \delta$ GC dimension, which will allow us to get a bound on the $\epsilon, \delta$ sample complexity.

\begin{proof}
    Suppose that for all $n \geq N$ and $\omega \in \Omega$, $\frac{\rademacher_{\hypoclass_\omega}(n)}{n} \leq \frac{\epsilon}{4}$.
    Now fix $n \geq N + \frac{8}{\epsilon^2} \log \frac{1}{\delta}$, and observe that $\frac{\rademacher_{\hypoclass_\omega}(n)}{n} \leq \frac{\epsilon}{4}$ still holds for all $\omega \in \Omega$.
    
    Then by Theorem \ref{thm:exp_mean_width}, we see that $n$ is also large enough that $\frac{\rademacher_{\expclass\classfamily}(n)}{n} \leq \frac{\epsilon}{4}$. Then setting
    $\gamma = \epsilon - \frac{2 \rademacher_{\expclass\classfamily}(n)}{n}$,
    our assumption implies that $\gamma \geq \frac{\epsilon}{2}$. Plugging in
    $\gamma$ for $\delta$ in Fact \ref{fact:rademacher}
    we have
    $GC_{\expclass\classfamily}\left(\epsilon,\exp\left(-\frac{n\gamma^2}{2}\right)\right) \leq n$.
    As $\gamma \geq \frac{\epsilon}{2}$ and $n \geq \frac{8}{\epsilon^2}\log\frac{1}{\delta}$,
    we have
    $$\exp\left(-\frac{n\gamma^2}{2}\right) \leq
    \exp\left(-\frac{n\epsilon^2}{8}\right)
    \leq \delta$$
    Thus the conclusion holds.
\end{proof}

\subsection{Proof of Theorem \ref{thm:randomvarclass} in the concept class case}
We now apply the lemma on pushing Rademacher mean width through an expectation in the context of a concept class.

In this setting, we can estimate $\rademacher_\conceptclass(n)$ in terms of VC-dimension, defining $\rademacher_\conceptclass(n)$
\begin{fact}[{\citet[Lemma 4.14 and Equation 4.24]{wainwright}}] \label{fact:rademachervc}
Assume that $\conceptclass$ is a concept class with VC-dimension at most $d_\conceptclass$. Then for $n \geq 1$,
$$\rademacher_\conceptclass(n) \leq 2\cdot \sqrt{d_\conceptclass \cdot n \cdot \log (n + 1)}.$$
\end{fact}
This fact will give us the bound on Rademacher mean width that we can plug in to  the ``pushing through expectation lemma'', Lemma \ref{lem:gc}.

Recall that $\classfamily = (\conceptclass_\omega : \omega \in \Omega)$ is a measurable family where each $\conceptclass_\omega$ is a $\{0,1\}$-valued class (that is, a concept class) 
with VC-dimension at most $d$.
Putting Fact \ref{fact:rademachervc} together with Lemma \ref{lem:gc}, we see that to bound the GC-dimension of the expectation class $\expclass\classfamily$, 
it suffices to find $N$ large enough that for $n \geq N$, $2\sqrt{\frac{d \log (n + 1)}{n}} \leq \frac{\epsilon}{4}$, and add $\frac{1}{\epsilon^2}\log \frac{1}{\delta}$.
This inequality is equivalent to
$$\frac{\log(n + 1)}{n} \leq \frac{\epsilon^2}{64 d},$$
which is guaranteed by 
$$\frac{\log n}{n} \leq \frac{\epsilon^2}{64 d \log 2},$$
where we call the constant on the right $\gamma$, noting that we can assume $\epsilon \leq 1$ and thus $\gamma < e^{-1}$.
Because the function on the left is decreasing for $n > e$, it suffices to find some $N$ for which this inequality holds. We try $N = C\gamma^{-1}\log\gamma^{-1}$, and see that
$$\frac{\log N}{N} = \frac{\log C + \log\gamma + \log\log\gamma^{-1}}{C\gamma\log\gamma^{-1}}
\leq \gamma\left(\frac{\log C + 2\log \gamma^{-1}}{C \log \gamma^{-1}}\right) \leq \gamma\left(\frac{\log C + 2}{C}\right),$$
using the fact that $\log\gamma^{-1} \geq 1$.
For sufficiently large $C$ (independent of $\gamma$), this is at most $\gamma$ as desired.
Thus
$$N = O\left(\gamma^{-1}\log\gamma^{-1}\right)
= O\left(\frac{d}{\epsilon^2}\log \frac{d}{\epsilon^2}\right) = O\left(\frac{d}{\epsilon^2}\log \frac{d}{\epsilon}\right),$$
and
$$GC_{\expclass\classfamily}(\epsilon,\delta) \leq N + \frac{8}{\epsilon^2}\log\frac{1}{\delta}
= O\left(\frac{1}{\epsilon^2}\left(d\log\frac{d}{\epsilon} + \log \frac{1}{\delta}\right)\right).$$

By Fact \ref{fact:gc_sample}, this completes the proof of Theorem \ref{thm:randomvarclass} in the case of concept classes.

\subsection{Extension to the real-valued case}\label{subsec:fat_shatter_GC}

We now extend to get sample complexity bounds for random objects, but where we start with a measurable family of hypothesis classes. Our aim will be:

\begin{theorem}
    For any $\epsilon,\delta > 0$,
    if $\classfamily$ is a measurable family of hypothesis classes such that each class $\hypoclass \in \classfamily$ has $\frac{\epsilon}{50}$ fat-shattering dimension at most $d$, the sample complexity of agnostic PAC learning on $\expclass\classfamily$ is bounded by:
    $$O\left(\frac{d}{\epsilon^4}\log^2\frac{d}{\epsilon} + \frac{1}{\epsilon^2}\log \frac{1}{\delta}\right).$$
\end{theorem}

The challenge will be in getting the required linear bound on Rademacher mean widths, so that
we can apply the lemma on pushing mean width through an expectation, Lemma \ref{lem:gc}.

We will use covering numbers.
The $\ell_p$ norm on $\reals^n$ is $\left(\sum_{i = 1}^n x_i^p \right)^{1/p}$, while $\ell_\infty$ is $\max_{i = 1}^n x_i$.

\begin{definition}
    For $A \subseteq \reals^n$, $\gamma > 0$, and $1 \leq p \leq \infty$, we let
$\mathcal{N}_p(\gamma,A)$, the \emph{$\gamma$ covering number of $A$}, be the minimum number of $\gamma$-balls in the $\ell_p$-metric that cover $A$.

We also let $\mathcal{N}_p(\gamma, \hypoclass, n)$ denote $\sup_{\bar x \in \rangespace^n}\mathcal{N}_p(\gamma,\hypoclass(\bar x,\paramspace)),$ with $\hypoclass(\bar x,\paramspace)\subseteq [0,1]^n$ defined as in Definition \ref{def:meanwidth}.
\end{definition}

We have defined these for arbitrary $p$, but from now we will use only $p = 2$ and $p = \infty$. The relevant relation between them is that for all $x \in \reals^n$, we have $|x|_2 \leq \sqrt{n}|x|_\infty$, so for any $A \subseteq \reals^n$, $\mathcal{N}_2(\gamma\sqrt{n}, A) \leq \mathcal{N}_\infty(\gamma, A)$, and for any $\hypoclass$ and $n$,
$$\mathcal{N}_2(\gamma\sqrt{n},\hypoclass,n) \leq \mathcal{N}_\infty(\gamma,\hypoclass,n).$$

We can bound covering numbers using fat-shattering:

\begin{fact}[{From the proof of \cite[Lemma 3.5]{alon}}]\label{fact_alon35}
Let $\hypoclass$ be a hypothesis class on $\rangespace$ parameterized by $\paramspace$. Let $d$ be the $\frac{\gamma}{4}$ fat-shattering dimension of $\hypoclass$.
Then
$$\mathcal{N}_\infty(\gamma,\hypoclass,n) \leq 2\left(\frac{4n}{\gamma^2}\right)^{d\log(2en/d\gamma)}.$$
\end{fact}
Here $e$ is the base of the natural logarithm.

To connect covering numbers to Rademacher mean width, we pass through another width notion, \emph{Gaussian mean width}.

\begin{definition}
    Let $\beta = (\beta_1,\dots,\beta_n)$, where the $\sigma_i$s are independent Gaussian variables with distribution $N(0,1)$. We define the \emph{Gaussian mean width}, denoted $\width_\gaussianmean(A)$, as $w(A,\beta)$, where $$\setwidth(A,\beta) = \mathbb{E}_{\beta}\left[h_A(\vec b)\right]$$ as in the definition of Rademacher mean width.
\end{definition}

We can easily relate Gaussian to Rademacher mean width, using the following fact:
\begin{fact}[{\citet[Exercise 5.5]{wainwright}}]\label{fact:rad_gauss}
    For any $A \subseteq [0,1]^n$, 
    $$\width_\rademacher(A) \leq \sqrt{\frac{\pi}{2}}w_\gaussianmean(A) \leq 2\sqrt{\log n}\,\width_\rademacher(A).$$
\end{fact}

We will only use the first of these two inequalities, but together they show that Gaussian and Rademacher mean widths are closely connected.
Covering numbers allow us to estimate how Gaussian mean width, and thus also Rademacher mean width, grows with dimension:
\begin{fact}[{\citet[Equation 5.36]{wainwright}}]\label{fact:gauss_cover}
    For $A \subseteq \reals^n$ with $\ell_2$-diameter at most $D$, and $0 \leq \gamma \leq D$,
    $$\width_\gaussianmean(A) \leq \gamma \sqrt{n} + 2D\sqrt{\log \mathcal{N}_2(\gamma, A)}.$$
\end{fact}
We will concern ourselves with $A \subseteq [0,1]^n$, so $D \leq \sqrt{n}$. Thus for $0 \leq \gamma \leq 1$, we plug in $\gamma\sqrt{n} \leq D$, and get 
$$w_\gaussianmean(A) \leq \gamma n + 2\sqrt{n\log \mathcal{N}_2(\gamma\sqrt{n}, A)}.$$

Combining the previous two facts gives us a straightforward way to relate covering numbers to Rademacher mean width:

\begin{corollary}\label{cor:covering-rademacher}
    Let $A \subseteq [0,1]^n$ and let $\gamma \in [0,1]$. Then
    $$\width_\rademacher(A) \leq \sqrt{\frac{\pi}{2}}\left(\gamma n + 2\sqrt{n \log \mathcal{N}_2(\gamma n,A)}\right)\leq \sqrt{\frac{\pi}{2}}\left(\gamma n + 2\sqrt{n \log \mathcal{N}_\infty(\gamma,A)}\right).$$
\end{corollary}
We now prove the remainder of Theorem \ref{thm:randomvarclass}, using the following bound on Glivenko-Cantelli dimension:
\begin{theorem}\label{thm:randomvarclassgc}
For any $\epsilon,\delta > 0$,
    if $\classfamily$ is a  measurable family of hypothesis classes such that each class $\hypoclass \in \classfamily$ has $\frac{\epsilon}{50}$ fat-shattering dimension at most $d$, then
    $$GC_{\expclass\classfamily}(\epsilon,\delta) = O\left(\frac{d}{\epsilon^4}\log^2\frac{d}{\epsilon} + \frac{1}{\epsilon^2}\log \frac{1}{\delta}\right)$$.
\end{theorem}
\begin{proof}
    By Lemma \ref{lem:gc}, it suffices to show that there is a constant $C > 0$ such that if $$n \geq C\frac{d}{\epsilon^4}\log^2\frac{d}{\epsilon},$$
    then
    $$\frac{\rademacher_\hypoclass(n)}{n} \leq \frac{\epsilon}{4}.$$
    As an intermediate bound, we can use Corollary \ref{cor:covering-rademacher} to bound the Rademacher complexity in terms of covering numbers, using $\gamma = \frac{\epsilon}{\sqrt{32\pi}}$:
    \begin{align*}
        \frac{\rademacher_\hypoclass(n)}{n}
        & = \sup_{\bar x \in \rangespace^n}\frac{w_\rademacher(\hypoclass(\bar x,\paramspace))}{n} \\
        &\leq \sup_{\bar x \in \rangespace^n} \sqrt{\frac{\pi}{2}}\left(\gamma + 2\sqrt{\frac{\log \mathcal{N}_\infty(\gamma,\hypoclass(\bar x,\paramspace))}{n}}\right)\\
        & \leq \sqrt{\frac{\pi}{2}}\left(\gamma + 2\sqrt{\frac{\log \mathcal{N}_\infty(\gamma,\hypoclass,n)}{n}}\right)\\
        &= \frac{\epsilon}{8} + \sqrt{\frac{2\pi\log \mathcal{N}_\infty(\gamma,\hypoclass,n)}{n}}.
    \end{align*}
    Thus it suffices to show that for suitably large $n$,
    $$\sqrt{\frac{2\pi\log \mathcal{N}_\infty(\gamma,\hypoclass,n)}{n}} \leq \frac{\epsilon}{8},$$
    or equivalently,
    $$\frac{\log \mathcal{N}_\infty(\gamma,\hypoclass,n)}{n} \leq \frac{\epsilon^2}{128\pi}.$$
    
    By Fact \ref{fact_alon35}, we see that for all $n$, 
    $$\log \mathcal{N}_\infty(\gamma,\hypoclass,n) = O\left(d\log\left(\frac{4n}{\gamma^2}\right)\log\left(\frac{2en}{d\gamma}\right)\right) = O\left(d\log^2\left(\frac{n}{\gamma^2}\right)\right) = O\left(d\log^2\left(\frac{32\pi n}{\epsilon^2}\right)\right).$$
    Now let $D$ be the constant of this inequality, so that for all $n$,
    $$\log \mathcal{N}_\infty(\gamma,\hypoclass,n)\leq Dd\log^2\left(\frac{32\pi n}{\epsilon^2}\right).$$
    It now suffices to show for suitably large $n$ that 
    $$\frac{Dd\log^2\left(\frac{32\pi n}{\epsilon^2}\right)}{n} \leq \frac{\epsilon^2}{128\pi}.$$
    Setting $a = \frac{32\pi}{\epsilon^2}$ and $b = \frac{\epsilon^2}{128\pi D d},$ we may assume that $a\geq 1$ and $\log(ab^{-1}) > 1$.
    We can restate our desired inequality as
    $$\frac{\log^2(a n)}{bn} \leq 1,$$
    and for some $C > 0$, we see that if $n \geq Cab^{-1}\log^2(ab^{-1})$, then
    \begin{align*}
   \frac{\log^2(an)}{bn}\leq
    \frac{\left(\log C + \log(a^2b^{-1}) + \log\log^2(ab^{-1})\right)^2}{bC(ab^{-1})\log^2(ab^{-1})}
    \leq \frac{\left(\log C + 4\log(a b^{-1})\right)^2}{aC \log^2(ab^{-1})}\\
    \leq \frac{(\log C/\log(ab^{-1}) + 4)^2}{C}\leq \frac{(\log C + 4)^2}{C},
     \end{align*}
    and it is clear that this bound is at most $1$ for large $C$.

    We then see that 
    $$Cab^{-1}\log^2(ab^{-1})
    = O\left(\frac{d}{\epsilon^4}\log^2 \frac{d}{\epsilon^4}\right)
    = O\left(\frac{d}{\epsilon^4}\log^2 \frac{d}{\epsilon}\right),$$
    which completes the proof.
\end{proof}

Theorem \ref{thm:randomvarclass} for general real-valued hypothesis classes follows from Theorem \ref{thm:randomvarclassgc} using Fact \ref{fact:gc_sample}.

\subsection{Simpler arguments for agnostic PAC learnability of the distribution function class} \label{subsec:sigmod22better}

A much simpler argument is available that provides bounds for the distribution class or dual distribution class, but which does not extend to provide a bound for general measurable families as in Theorem \ref{thm:randomvarclass}. We explain the idea for the dual distribution class formed over a concept class. We know that if a concept class is agnostic PAC learnable, then so is its dual class, where the concepts are given by elements $x$ of the range space. We can then conclude by routine calculation with dimensions,  that for each $k$, the functions on the parameter space given by normalized sums of elements $\frac{x_1 + \ldots  + x_k}{k}$ is agnostic PAC learnable. But by a basic result in statistical learning theory, the fact that our original concept class is PAC learnable means we can approximate each function in the dual distribution class arbitrarily closely by normalized sums.

Recall that in the special case of a concept class, the result
for the dual distribution class had been proven in prior work \citep{sigmod22}. The simpler
proof we present here, like the analytic proof that goes via the expectation class  of a measurable family, improves on the bound given in prior work.
We now explain the idea of this alternative approach.

Fix a  concept class $\conceptclass$ on $\rangespace$ indexed by parameter set $\paramspace$ such that $\conceptclass$ has finite VC dimension $d_{\conceptclass}$ and dual VC dimension $d^*_\conceptclass$.

We let $\chi^\conceptclass_{x}$ be the dual family of characteristic functions: the family of functions on $\paramspace$, indexed by elements of $\rangespace$, given by 
$\chi^\conceptclass_{x}(c)= 1$ if $x \in \conceptclass_{c}$,  $0$ otherwise.

Thus for any $\gamma$, 
the $\gamma$ fat-shattering dimension of this family is just the dual VC dimension
$d^*_\conceptclass$. This is for every $1>\gamma>0$, independent of $\gamma$.

We let $\EAvg_m(v_1 \ldots v_m)$ be the average of $v_1 \ldots v_m$, thus
$\EAvg_m$ is a function from $\reals^m$ to $\reals$.

Let $\chi^\conceptclass_m$ be the composed class, indexed by $x_1 \ldots x_m \in \rangespace$, 
with each function taking  $c \in \paramspace$ to
\[
\EAvg_m(\chi^\conceptclass_{x_1}(c) \ldots \chi^\conceptclass_{x_m}(c))
\]

This fits the composition framework in Theorem 1 of \citep{aggregation}.

\begin{fact} \citet{aggregation}
The $\gamma$ fat-shattering dimension of the composed class 
$\chi^\conceptclass_m$ is bounded by:

\[
25  \cdot D_\gamma \log^2 (90 \cdot D_\gamma )
\]
\end{fact}

where $D_\gamma$ here is the sum from $1$ to $m$ of the $\gamma$ fat-shattering dimension of the constituent class, which in this case is just
$m \cdot d^*_{\conceptclass}$.

Thus the $\gamma$ fat-shattering dimension of the  class  $\chi^\conceptclass_m$ is bounded by:
\[
2 5 \cdot m \cdot d^*_{\conceptclass} \cdot  [\log(90) + \log m + \log d^*_{\conceptclass} ] ^2 
\]

Call this $J(m, d^*_{\conceptclass})$.


We now want to control the relationship between arbitrary measures and averages. 

Fix $\gamma>0$. Recall  that $\dualmeasureclass{\conceptclass}$ is the dual distribution function class for the concept class  $\conceptclass$. Suppose the $\gamma$ fat-shattering dimension of $\dualmeasureclass{\conceptclass}$ were greater than
or equal to $k'$.
Let $n_k$ be large enough that every measure has a $\frac{\gamma}{2}$-approximation of size $n_k$: that is, there is a set of size $n_k$ elements of $\rangespace$ such that for any
$\vec c$, the percentage of elements in the set satisfying $\conceptclass_{\vec c}$ is within $\frac{\gamma}{2}$ of the measure of the set.
Then the $\frac{\gamma}{2}$ fat-shattering dimension of the class $\chi^\conceptclass_{n_k}$ 
would be above $k'$.

Thus $k' \leq J(n_k, d^*_{\conceptclass})$, or restated, the $\gamma$ fat-shattering dimension of $\dualmeasureclass{\conceptclass}$ is bounded by 
$J(n_k, d^*_\conceptclass)$.

We can use the following fact, which can be found in standard learning theory texts: see, e.g. \citep{samplecomplexitylearning} or for an exposition \citep{pillowmath} Theorem 4.2.

\begin{fact} 
There is a constant $L$, such that
if we let $n_k$ be such that $L \cdot \frac{ \sqrt{d_\conceptclass} }{ \sqrt{n_k} } \leq \frac{\gamma}{2}$, then
 every measure has a $\frac{\gamma}{2}$-approximation of size $n_k$.
\end{fact}


Thus a bound for $n_k$ can be taken to be:
\[
\frac{L' \cdot d_\conceptclass}{\gamma^2}
\]
for another universal constant $L'$.

Plugging into the bound for $J(n_k, d^*_\conceptclass)$, and ignoring log factors and universal constants, we get a bound on the $\gamma$ fat-shattering dimension for $\dualmeasureclass{\conceptclass}$ on the order of:
\[
  \frac{ (d_\conceptclass) \cdot d^*_\conceptclass }{ \gamma^2 }
\]

Plugging into the sample complexity bound of Fact \ref{fact:fatshatterandlearn}
we get the sample complexity of agnostic PAC learning $\measureclass{\conceptclass}$ is bounded by:

\[
O\left(\frac{1}{\epsilon^2} \cdot \left[\frac{ d_\conceptclass \cdot d^*_\conceptclass }{ (\frac{\epsilon}{9})^2 } \log^2\left( \frac{1}{\epsilon}\right) + \log\left(\frac{1}{\delta}\right)\right] \right)
\]

Recall from the body of the paper that in \citep{sigmod22} the dual distribution class over a concept class was considered, and a sample complexity bound of 
$\tilde{O}(\frac{1} {\epsilon^{\lambda+1}})$ was obtained, where $\tilde{O}$ indicates that we drop terms that are polylogarithmic in $\frac{1}{\epsilon}$,
$\frac{1}{\delta}$ for constant $\lambda$.
In contrast, in our bound above,  we have a constant in the exponent of $\epsilon$ in the denominator, rather than the VC-dimension.

We now show that this approach generalizes easily from concept classes to function classes. That is, we give an alternative proof of the preservation of agnostic PAC learnability, without going through the expectation class of a measurable family. We do not compute sample complexity bounds explicitly for this alternative proof, but they are similar to those given above.

\begin{theorem}\label{thm:pac_dist_alternate} Let $\hypoclass$ be a class of functions over $X$, indexed by $Y$.
Suppose $\hypoclass$ is agnostic PAC learnable (equivalently has finite $\gamma$ fat-shattering dimension for 
each $\gamma$) then the same holds for the dual distribution function class (and the distribution function class) of $\hypoclass$.
\end{theorem}

We let $\hypoclass^*$ be the dual family. Now this is a class of functions over $Y$, indexed by $X$. Since agnostic PAC learning for real-valued functions is closed under dualization by Fact \ref{fact:pacdual},
for any $\gamma$,  the $\gamma$ fat-shattering dimension of this family is also finite.

As before, let $\EAvg_m(v_1 \ldots v_m)$ be the average of $v_1 \ldots v_m$.
Let $\EAvg_m(\hypoclass^*)$ be the composed class, indexed by $x_1 \ldots x_m$ in $\rangespace^m$
taking $c$ to
$\EAvg_m(h_{x_1}(c) \ldots h_{x_m}(c))$.

This again fits the composition framework in Theorem 1 of \citep{aggregation}.
From the theorem we have the $\gamma$ fat-shattering dimension of the composed class 
 is bounded by:

\[
25 \cdot m \cdot D_\gamma \cdot \log^2 (90 \cdot D_\gamma )
\]

where $D_\gamma$ here is the sum from $1$ to $m$ of the $\gamma$ fat-shattering dimension of the constituent class, which in this case is just
$m \cdot \fatshatter_\gamma(\hypoclass^*)$.

Thus the $\gamma$ fat-shattering dimension of the  class  $\EAvg_m(\hypoclass^*)$ is bounded by:
\[
2 5 \cdot (m^2) \cdot \fatshatter_\gamma(\hypoclass^*) [\log90 +\log m + \log \fatshatter_\gamma(\hypoclass^*) ] ^2 
\]

Call this $J(m, \fatshatter_\gamma(\hypoclass^*))$.

Fix $\gamma$, and again let $\dualmeasureclass{\hypoclass}$ be the dual distribution function class for $\hypoclass$.
Suppose the $\gamma$ fat-shattering dimension of $\dualmeasureclass{\hypoclass}$ were greater than
or equal to $k'$.
This time let $n_k$ be large enough that every distribution has an $n_k$ sized $\frac{\gamma}{2}$-approximation, in the sense  that there is an $n_k$ sized tuple $(x_1,\dots,x_{n_k})$ in $\rangespace$ such that for any $h \in \hypoclass$, the average of $h$ over the elements of the tuple is within $\frac{\gamma}{2}$ of the mean of $h$ - that is,
$$\left|\frac{1}{n_k}\sum_{i = 1}^{n_k}h(x_i) - \mathbb{E}[h]\right| \leq \frac{\gamma}{2}.$$
Then the $\frac{\gamma}{2}$ fat-shattering dimension of the class $\chi^\phi_{n_k}$ 
would be above $k'$.

Thus $k' \leq J(n_k, \fatshatter_\gamma(\hypoclass^*))$.
Restated: the $\gamma$ fat-shattering dimension of $\dualmeasureclass{\hypoclass}$ is bounded by 
$J(n_k, \fatshatter_\gamma(\hypoclass^*))$.

It suffices to take $n_k$ large enough that, for some fixed $0 <\delta < 1$, $GC_\hypoclass\left(\frac{\gamma}{2},\delta\right)\leq n_k$, because if this is the case, the probability that a randomly-selected tuple is a $\frac{\gamma}{2}$-approximation is at least $1 - \delta$. Thus, fixing $\delta$, we may use \begin{align*}
n_k = 
O\left({\frac{1}{\epsilon^2}}\cdot \pshatter_{\frac{\epsilon}{9}}\left(\hypoclass\right) \cdot  \log^2\left( \frac{1}{\epsilon} \right) \right).
\end{align*}

\subsection{The realizable case for PAC learning} \label{subsec:realizablesupervised}

The previous subsections showed preservation of agnostic PAC learnability for the distribution class construction. We now show by contrast that the distribution class constructions do not preserve PAC learnability in the realizable case.

\begin{proposition} \label{prop:realizablesupnotpreserved} There is a hypothesis class $\hypoclass$ that is realizable PAC learnable,
but the distribution function class and dual distribution function class based on $\hypoclass$
are not realizable PAC learnable.
\end{proposition}

In the proof, we let $\hypoclass_0$ be the following slight modification of a class from \cite[Example 1]{realizableonline}.
Let $X$ be a set, partitioned into nonempty pieces $X_0,X_1,\dots$, with characteristic functions $\chi_{X_i}$.
    Let $B \subseteq \{0,1\}^\nats$ consist of all sequences of bits with only finitely many ones, and let $\hypoclass_0 = \{h_b : b \in B\}$, where
    if $b = (b_0,b_1,\dots),$ then
    $$h_b(x) = \frac{3}{4}\sum_{i = 0}^\infty b_i \cdot \chi_{X_{i}}(x) + \frac{1}{8}\sum_{i = 0}^\infty b_i \cdot  2^{-i}.$$

The class $\hypoclass_0$  has infinite $\gamma$ fat-shattering dimension for all $\gamma < \frac{1}{4}$, so is not agnostic PAC learnable. But it is realizable learnable  -- in fact, it is learnable %
from one sample, as for any $x$ and distinct $h,h' \in \hypoclass_0$, $h(x) \neq h'(x)$.

We can easily see that the dual distribution class of this class is not realizable PAC learnable:
\begin{lemma}
    The dual class of $\hypoclass_0$ is not PAC learnable in the realizable case.
    Hence the dual distribution class is not  PAC learnable in the realizable case.
\end{lemma}
\begin{proof}
A dual class element is given by an $x_0$ in the range space, with any other element in the same partition $X_i$ as $x_0$ inducing the same function.
If we see samples $\vec b_1 \ldots \vec b_n$, with value at most $\frac{1}{4}$, we will be able to exclude some partition elements $X_i$, but we will have infinitely many $X_i$ possible.
\end{proof}

We now show the same thing for the distribution class.

\begin{lemma}\label{lem:notrealizablepac}
    The distribution class  of $\hypoclass_0$ is  not PAC learnable in the realizable case.
    In fact, we may simply look at the class on $X$ of ``two choice distributions'', consisting of all hypotheses
    $\lambda \cdot  h_b + (1-\lambda) \cdot h_{b'}$
    for rational $\lambda \in [0,1]$ with $b,b' \in B$.
\end{lemma}
\begin{proof}
    We prove this by showing that this new class has infinite $\frac{1}{8}$-graph dimension, as defined in \citep{realizableonline}, which we now review.
    Fix a natural number $n$. For our purposes, we only need to know when a class has $\frac{1}{8}$-graph dimension at least $n$. 
    The definition of graph dimension states that this holds when we have $x_0,\dots,x_{n - 1} \in \rangespace$, $f_1,\dots,f_{n - 1} \in [0,1]$, and for each $\beta \in \{0,1\}^n$, a hypothesis $h'_\beta$ such that 
    \begin{itemize}
        \item $h'_\beta(x_i) = f_i$ when $\beta_i = 0$
        \item $|h'_\beta(x_i) - f_i| > \frac{1}{8}$ when $\beta_i = 1$.
    \end{itemize}

    Note that, unlike with fat-shattering,  we have an equality for the zero values of a branch and an inequality for the one value. Theorem 1 of \citep{realizableonline} shows that an infinite $\frac{1}{8}$-graph 
    dimension  --- that is, having such witnesses for each $n$ --- implies that  the class is not realizable PAC learnable. 
    To find the required witnesses, let $x_i \in X_i$ for $i < n$, let $f_i = \frac{1}{2}$ for each $i$, and let $1_n = (1,1,\dots,1,0,0,\dots) \in \{0,1\}^\N$ be such that $(1_n)_i = 1$ exactly when $i < n$. For each $\beta \in \{0,1\}^n$, let $\beta' \in \{0,1\}^\N$ be the sequence extending $\beta$ with zeros, that is, $\beta'_i = 0$ for $i \geq n$. For $b \in \{0,1\}^\N$, let
    $$c_b = \frac{1}{8}\sum_{i = 0}^\infty b_i\cdot 2^{-i},$$
    so that for any $x_i \in \rangespace_i,$
    $$h_b(x_i) = \frac{3}{4} b_i + c_b.$$
    We note that for all $b$,
    $$0 \leq c_b \leq \frac{1}{8}\sum_{i = 0}^\infty 2^{-i} = \frac{1}{4},$$
    and if there are only finitely many $i$ such that $b_i = 1$, as is the case if $b = \beta'$ for some $\beta \in \{0,1\}^n$ or $b = 1_n$ for some $n$, then $c_b$ is rational.
  
    We claim that there is some $\lambda \in [0,1] \cap \mathbb{Q}$ such that letting $h'_\beta = \lambda \cdot h_{\beta'} + (1-\lambda) \cdot h_{1_n}$ for each $\beta$, we will obtain the two required properties above.
    We find that for $x_i$ with $i < n$, we have, for each $\beta$
    \begin{align*}
        h'_\beta(x_i) &= \lambda h_{\beta'}(x_i) + (1 - \lambda)h_{1_n}(x_i)\\
    &= \lambda\left(\frac{3}{4}\beta_i + c_{\beta'}\right) + (1-\lambda)\left(\frac{3}{4} + c_{1_n}\right)
    \end{align*}
    Again fixing $\beta$, we observe that 
    $c_{\beta'} < \frac{1}{2} < \frac{3}{4} + c_{1_n}$ and both $c_{\beta'}$ and $c_{1_n}$ are rational. Thus we can choose $\lambda \in [0,1] \cap \mathbb{Q}$ with
    $$\lambda\left(c_{\beta'}\right) + (1-\lambda)\left(\frac{3}{4} + c_{1_n}\right) = \frac{1}{2}.$$

    Then for each $i$, if $\beta_i = 0$, we have
    $$h'_\beta(x_i) = \lambda\left(c_{\beta'}\right) + (1-\lambda)\left(\frac{3}{4} + c_{1_n}\right) = \frac{1}{2} = f_i,$$
    and if $\beta_i = 1$, we have both $h_{\beta'}(x_i),h_{1_n}(x_i) \geq \frac{3}{4}$, so $h'_\beta(x_i) = \lambda \cdot h_{\beta'}(x_i) + (1-\lambda) \cdot h_1(x_i) \geq \frac{3}{4}$. Thus in particular $|h'_\beta(x_i)-f_i| > \frac{1}{8}$ as required.
\end{proof}

We now show that we can ``fix'' these anomalies by dealing not with a single hypothesis classes, but a family with reasonable closure properties. If we  consider classes closed under composition with continuous bijections $[0,1] \to [0,1]$, then there is no difference between the realizable and agnostic case, and we do have preservation under moving to statistical classes.
\begin{proposition} \label{prop:clsupervisedregression}
    Let $\classfamily$ be a family of hypothesis classes on range space $\rangespace$ parametrized by $\paramspace$ closed under applying continuous bijections: for any $\hypoclass \in \classfamily$, the composition of $\hypoclass$ with a continuous bijection $f$ is another class in $\classfamily$.
    If $\classfamily$ contains a hypothesis class
 with infinite $\gamma$ fat-shattering dimension for some $\gamma > 0$, there is another hypothesis class in $\classfamily$ which  is not realizable PAC learnable. Thus in particular if every hypothesis class in $\classfamily$ is realizable PAC learnable, then every class in $\classfamily$ is agnostic PAC learnable, and thus every distribution class or dual distribution class arising from a class in $\classfamily$ is agnostic PAC learnable and realizable PAC learnable.
\end{proposition}

Before beginning the proof of the proposition, we recall a basic result on the fat shattering dimension.
If a class has infinite $\gamma$ fat-shattering dimension for some $\gamma > 0$, this means we get arbitrary large powersets that we can capture with a $\gamma$-gap. The following result states that realize these counterexamples with a uniform choice of number $r<s$, where $s - r$ is bounded below by a function of $\gamma$:

\begin{fact}[{\citet[Thm. 4.2]{alon}}] \label{fact:highpdimension}
For every $\gamma > 0$, there is some $\beta > 0$ such that the following holds:

Consider a real-valued hypothesis class $\hypoclass$ consisting of functions $h_y$ for $y$ ranging over parameter set. Suppose that $\hypoclass$ has infinite $\gamma$ fat-shattering dimension. Then
for every natural number $d$, there are also $0 \leq r < s \leq 1$ such that $s - r \geq \beta$ and there are $x_1,\dots,x_d \in \rangespace$ and $(y_b : b \in \{0,1\}^d)$ such that for each $b$ and $i$, if $b(i) = 0$, then $h(x_i;y_b)\leq r$, and if $b(i) = 1$, then $h(x_i;y_b) \geq s$.
\end{fact}

In the literature, this is sometimes phrased by saying that $\hypoclass$ has ``infinite $V_\beta$-dimension'' (where $V$ is for Vapnik).

We now begin the proof of Proposition \ref{prop:clsupervisedregression}:

\begin{proof}
Fix a family $\classfamily$ of hypothesis classes that is closed under applying continuous bijections.
Assume there is a class $\hypoclass \in \classfamily$ that is not agnostic PAC learnable, and we will prove that there is some $\hypoclass' \in \classfamily$ that is not realizable PAC learnable.

In particular, we assume that $\hypoclass$ has infinite $\gamma$ fat-shattering dimension for some $\gamma > 0$. Then by Fact \ref{fact:highpdimension}, there is some $\beta > 0$ such that
for every natural number $d$, there are also $0 \leq r < s \leq 1$ such that $s - r \geq \beta$ and there are $x_1,\dots,x_d \in \rangespace$ and $(y_b : b \in \{0,1\}^d)$ such that for each $b$ and $i$, if $b(i) = 0$, then $\hypoclass_{y_b}(x_i)\leq r$, and if $b(i) = 1$, then $\hypoclass_{y_b}(x_i) \geq s$.

Let $f : [0,1] \to [0,1]$ be a continuous bijection such that $f(r) = 0$ and $f(s) = 1$. Then consider the hypothesis class $\hypoclass' = \{h'_b : b \in \paramspace\}$ defined by $h'_b(x) = f(h_b(x))$.
Because $\classfamily$ is closed under applying continuous bijections, $\hypoclass' \in \classfamily$, and we will show that $\hypoclass$ is not realizable PAC learnable.

To do this, we show that the $\frac{1}{8}$-graph dimension of $\hypoclass'$ is infinite. (In fact, the $1-$graph dimension is.) For every $d$, there are there are $x_1,\dots,x_d \in \rangespace$ and $(y_b : b \in \{0,1\}^d)$ such that for each $b$ and $i$, if $b(i) = 0$, then $\hypo'_{y_b}(x_i) = 0$, and if $b(i) = 1$, then $\hypo_{y_b}(x_i) = 1$.

To show that the $\frac{1}{8}$-graph dimension is greater than $d$, by the characterization in the proof of Lemma \ref{lem:notrealizablepac}, it now suffices to define $f_1,\dots,f_{n - 1} \in [0,1]$ to all be $0$, which ensures that
    \begin{itemize}
        \item $\hypo'_{y_b}(x_i) = f_i$ when $b_i = 0$
        \item $|\hypo'_{y_b}(x_i) - f_i| > \frac{1}{8}$ when $b_i = 1$,
    \end{itemize}
because $|\hypo'_{y_b}(x_i) - f_i| = \hypo'_{y_b}(x_i)$.
\end{proof}

\section{Online learnability of statistical classes}\label{sec:online}
We will now be interested in getting an analog of Theorem \ref{thm:randomvarclass} for online learning, deriving bounds on learnability for statistical classes based on dimensions of the base class.
As we did in the agnostic PAC case, we will work via a broader result on measurable families. Recall from Fact \ref{fact:sfs_online} that agnostic PAc learnability is controlled by the sequential fat-shattering dimension of a hypothesis class.
We will calculate the following regret bound on the expectation class of a measurable family in terms of a bound on sequential fat-shattering dimension for the members in the family:

\begin{theorem} \label{thm:randomvaronline}
    For any $\epsilon,\delta > 0$,
    if $\classfamily$ is a measurable family of hypothesis classes such that each class $\hypoclass \in \classfamily$ has $\frac{\epsilon}{50}$ sequential fat-shattering dimension at most $d$, the minimax regret of online learning for $\expclass\classfamily$ over runs of length $n$ is at most
$$4 \cdot \gamma \cdot n + 12 \cdot (1-\gamma) \cdot \sqrt{d \cdot n \cdot \log\left(\frac{2 \cdot e \cdot n}{\gamma}\right)}.$$

If the classes in $\classfamily$ are $\{0,1\}$-valued and have Littlestone dimension at most $d$, the minimax regret of online learning for $\expclass\classfamily$ over runs of length $n$ is at most
$O\left(\sqrt{d \cdot n}\right)$.
\end{theorem}

Note that we can again apply Proposition \ref{prop:embeddingbottomline} to apply this to the distribution class of a fixed hypothesis class that has finite sequential fat-shattering dimensions. As finiteness of $\gamma$ sequential fat-shattering dimension for all $\gamma > 0$ is equivalent to sublinear minimax regret by Fact \ref{fact:sfs_online}, we have the following corollary:
\begin{corollary} \label{cor:preserved}
    If function class $\hypoclass$ is agnostic online learnable, in the sense of having sublinear minimax regret, then so are the distribution class and the dual distribution class.
\end{corollary}

For the corollary, we use the fact that given any class $\hypoclass$ of finite $\gamma$ sequential fat-shattering dimension $d$, every element of the the  parameter randomized family of $\hypoclass$ (see Definition \ref{def:paramrandomizedfamily})  has $\gamma$ sequential fat-shattering dimension at most $d$. So its expectation class,  $\hypoclass(\compatparamrv(\hypoclass))$, as in Definition \ref{def:paramrandomizedexpectationclass}, is agnostic online learnable. The distribution class then embeds into $\hypoclass(\compatparamrv(\hypoclass))$. To deal with the dual distribution class, we also use the fact that agnostic online learnability is closed under dualization by Proposition \ref{prop:onlinedual}.

\subsection{Proof of the main theorems on agnostic online learnability of statistical classes, with quantitative bounds}
We now present the proof of Theorem \ref{thm:randomvaronline}.

Recall that to push agnostic PAC learning from a base class into a statistical class, we went through notions of width. We will do something similar here.

To bound regret for online learning for the expectation class, we will adapt the framework of sequential Rademacher mean width developed in \citep{rakhlin1,rakhlin2}.
We will define it here, slightly amending the definition to better match our conventions for mean width.

\begin{definition}[Sequential Rademacher mean width]
For any $n$, let $\{1,-1\}^{< n} = \bigcup_{t = 0}^{n-1} \{1,-1\}^t$.
For each sequence $s = (s_1,\dots,s_n) \in \{1,-1\}^n$, let $v_s \in \reals^{\{1,-1\}^{< n}}$ be the vector such that for $1 \leq t \leq n$, $(v_s)_{(s_1,\dots,s_{t-1})}=s_{t + 1}$, while all other entries are 0. Then let $T_n = \{v_s : s \in \{1,-1\}^n\}$.
Recall the definition of mean width from Definition \ref{def:meanwidth}.
Define the \emph{sequential Rademacher mean width} of a set $A \subseteq \reals^{\{1,-1\}^{< n}}$, $w^S_\rademacher(A)$, to be $w(A,\beta)$ where $\beta$ is the uniform distribution on the $2^n$ elements of $T_n$.

For a hypothesis class $\hypoclass$ on $\rangespace$, we then define the \emph{sequential Rademacher mean width}, $\seqrademacher_f(n)$, to be $\sup_{\bar x \in \rangespace^{\{1,-1\}^{< n}}}w_\seqrademacher(\hypoclass(\bar x,\paramspace)$.
This definition coincides with the one from \citep{rakhlin1} except for the factor of $\frac{1}{n}$, although we call it a mean width instead of a complexity.
\end{definition}

Our argument will rely on the following bound, extracted from \cite[Proposition 9]{rakhlin2}. The differences between this fact and the statement in \citep{rakhlin2} are due to our slightly different definition and the fact that our function classes take values in $[0,1]$ instead of $[-1,1]$.
\begin{fact}[{See \cite[Proposition 9]{rakhlin2}}]\label{fact:meanwidthonline}
    Let $\hypoclass$ be a function class on $\rangespace$ taking values in $[0,1]$.
    The minimax regret of online learning for $\hypoclass$ on a run of length $n$ is at most
    $\seqrademacher_\hypoclass(n)$,
    which is in turn at most
    \begin{small}
    \begin{align*}
    \inf_\gamma 
    4 \cdot \gamma \cdot  n + 12 \cdot \sqrt{n} \cdot \\
    \int_\gamma^1 \sqrt{\seqfatshatter_\beta(\hypoclass) \cdot \log\left(\frac{2 \cdot e \cdot n}{\beta}\right)}\, \mathrm{d}\beta. 
    \end{align*}
    \end{small}
\end{fact}

As with Rademacher mean width, the advantage of this dimension is that we can push it through expectations. 

\begin{theorem}\label{thm:exp_seq_mean_width}
    Let $\classfamily = (\hypoclass_\omega : \omega \in \Omega)$ be a measurable family of hypothesis classes on $\rangespace$. 
    Then
    $$\seqrademacher_{\mathbb{E}[f]}(n) \leq \sup_\omega \seqrademacher_{\hypoclass_\omega}(n).$$
\end{theorem}


\begin{proof}
    Recall that the definition of sequential Rademacher mean width is the same as Rademacher mean width, except with a different probability distribution.

    For any $n$, the distribution $\beta$ we use on $\reals^{\{-1,1\}^{<n}}$ is the uniform distribution on a particular finite set $T_n$.
    Then the sequential Rademacher mean width of a set $A \subseteq \reals^{\{1,-1\}^{< n}}$, $w^S_\rademacher(A)$ is $w(A,\beta)$.
    For a function $f(x,y)$, the sequential Rademacher mean width was defined by $$\seqrademacher_f(n) = \sup_{\bar x \in \rangespace^{\{1,-1\}^{< n}}}w_\seqrademacher(\hypoclass(\bar x,\paramspace).$$

    This allows us to restate this theorem statement as showing that for each $n$,
    $$\sup_{\bar x \in \rangespace^{\{-1,1\}^{<n}}}w(\expclass\classfamily(\bar x,\paramspace),\beta)
    \leq \sup_\omega \sup_{\bar x \in \rangespace^{\{-1,1\}^{<n}}} w(\hypoclass_\omega(\bar x,\paramspace),\beta),$$
    and the proof of this is essentially identical to the proof of Theorem \ref{thm:exp_mean_width}, but with a new distribution $\beta$.
\end{proof}
From this, we can prove a bound on regret for the expectation class in terms of the sequential fat-shattering dimension, proving the first part of Theorem \ref{thm:randomvaronline}.
\begin{theorem}
    Let $\classfamily = (\hypoclass_\omega : \omega \in \Omega)$ be a measurable family of hypothesis classes on $\rangespace$. 
    
    The minimax regret of online learning for $\expclass\classfamily$ with $\gamma$ sequential fat-shattering dimension at most $d$ on a run of length $n$ is at most
    $$4 \cdot \gamma \cdot  n + 12 \cdot (1-\gamma) \cdot \sqrt{d \cdot n \cdot \log\left(\frac{2 \cdot e \cdot n}{\gamma}\right)}.$$
\end{theorem}
\begin{proof}
    By Theorem \ref{thm:exp_seq_mean_width}, any uniform (in $\omega$) bound on $\seqrademacher_{\hypoclass_\omega}(n)$ also applies to the expectation class, so regret for the expectation class is bounded by
    $$4 \cdot \gamma \cdot n + 12 \cdot \sqrt{n} \cdot \int_\gamma^1 \sqrt{\pshatter_\beta^S(\hypoclass) \cdot \log\left(\frac{2 \cdot e \cdot n}{\beta}\right)}\,\mathrm{d}\beta.$$
    Because the function in the integral is decreasing in $\beta$, we may bound it na\"ively by the value at $\gamma$.
\end{proof}

For a $\{0,1\}$-valued concept class, all sequential fat-shattering dimensions coincide with the Littlestone dimension, and the following improved bound holds:
\begin{fact}[{\citet[See Lemma 6.4 and Theorem 12.2]{alon_adversarial}}]
    If $\conceptclass$ is a $\{0,1\}$-valued concept class with Littlestone dimension at most $d$, then
    $$\seqrademacher_{\hypoclass}(n) = O(\sqrt{d \cdot n}).$$
\end{fact}
From this, we are able to conclude that the minimax regret for the expectation class of a family of concept classes with Littlestone dimension bounded by $d$ is also at most $O(\sqrt{d \cdot n})$, which proves the second part of Theorem \ref{thm:randomvaronline}.

\subsection{Preservation of online learnability in moving to statistical classes, via stability} \label{subsec:onlinelearningviastability}
Above we showed that online learnability is preserved in moving to statistical
classes.
We now give a variant without the quantitative bounds, but one which derives from prior results stated in the context of model theory. We can show that agnostic online learnability is equivalent to the notion of \emph{stability}, a notion originating in logic. We can then use prior results \citep{ibycontinuousrandom,ibyrandvar} on the preservation of stability in moving to an expectation class.

\begin{definition} \label{def:stableconcept}
A concept class $\conceptclass$ over $\rangespace$ parameterized by $\paramspace$ is \emph{stable} if 
there do not exist arbitrarily large sequences $a_i, b_i: 1 \leq n$ with $a_i \in \rangespace, b_i \in \paramspace$ such that $\forall i,j \leq n$ $a_i \in \concept_{b_j}$ if and only if $i<j$. 
\end{definition}
Roughly speaking stability says that $\conceptclass$ does not define arbitrarily large linear orders. This can be phrased as $\conceptclass$ having finite \emph{threshold dimension}, as it is called in \cite{threshold}, where this dimension is the largest size of a finite sequence of pairs $(a_i,b_i)$ with the above property.

Recall that for PAC learning of concept classes the critical dimension is VC dimension or NIP: not having arbitrarily large shattered sets.
Stability is the analogous dividing line for online learning, agnostic or realizable, in the case of concept classes:

\begin{fact} \label{fact:stableonlinelearnable} \citep{chasefreitag} 
A concept class is stable if and only if it is online learnable.
\end{fact}
Here we refer to learnability either in the realizable case or the agnostic case, which are equivalent for a concept class, as noted in the preliminaries. Both are equivalent to the class having finite Littlestone dimension by Fact \ref{fact:sfs_online}.

Thus far we are reviewing a connection between  stability and online learnability for concept classes, which is already known.
We now turn to real-valued classes. The notion of a stable hypothesis class generalizes to this setting, but now requires a real parameter. It can be defined in terms of either of the following dimensions:
\begin{definition}
    Let $\hypoclass$ be a hypothesis class on $\rangespace$.

    For any $\gamma > 0$ and any $d$, say that $\hypoclass$ has \emph{$\gamma-$threshold dimension at least} $d$ when there are $a_1,\dots,a_d \in \rangespace$, $\hypo_1,\dots, \hypo_d \in \hypoclass$ such that for all $i < j$,
    $$|\hypo_j(a_i) - \hypo_i(a_j)|\geq \gamma.$$

    For any $r < s$ and any $d$, say that $\hypoclass$ has \emph{$(r,s)-$threshold dimension at least} $d$ when there are $a_1,\dots,a_d \in \rangespace$, $\hypo_1,\dots, \hypo_d \in \hypoclass$ such that for all $i < j$,
    $\hypo_j(a_i) \leq r$ and $\hypo_i(a_j) \geq s$.

    Call a hypothesis class $\hypoclass$ $\gamma$-\emph{stable} when it has finite $\gamma$-threshold dimension, and call $\hypoclass$ $(r,s)$-\emph{stable} when it has finite $(r,s)$-threshold dimension.
\end{definition}

These dimensions give two equivalent definitions of overall stability for a hypothesis class, as shown in the model-theoretic context in \cite[Section 7]{localstab}:
\begin{fact}[{\citet[Section 7]{localstab}}]\label{fact:clstableanddefinability} If $\hypoclass$ is a hypothesis class, the following are equivalent:
    \begin{itemize}
        \item For every $\gamma > 0$, $\hypoclass$ is $\gamma$-stable.
        \item For every $r < s$, $\hypoclass$ has is $(r,s)-$stable.
    \end{itemize}
\end{fact}

We call such a class \emph{stable}. Roughly speaking stability says that we cannot use gaps in function values, discretized up to some $\gamma$, to define arbitrarily large linear orders.

We will give a slightly more explicit version of this equivalence, using largely the same proof as in \cite{localstab}, but with only finitary Ramsey theory, in order to define stability uniformly over a family of hypothesis classes.

\begin{definition}
    If $\classfamily$ is a family of hypothesis classes on $\rangespace$ parameterized by $\paramspace$, and $\gamma > 0$, say that $\classfamily$ is \emph{uniformly $\gamma$-stable} when there is some $d$ such that every $\hypoclass \in \classfamily$ has $\gamma-$threshold dimension at most $d$.
    
    For $r < s$, say that $\classfamily$ is \emph{uniformly $(r,s)$-stable} when there is some $d$ such that every $\hypoclass \in \classfamily$ has $(r,s)-$threshold dimension at most $d$.
\end{definition}

These two notions are closely related:
\begin{lemma} \label{lem:stablenotions}
    For any $\gamma > 0$ and $d$, if $\hypoclass$ has $\gamma$-threshold dimension less than $d$, then for all $r, s$ with $r + \gamma \leq s$, $\hypoclass$ has $(r,s)$-threshold dimension less than $d$.

    For any $0 < \delta < \gamma$ and $d$, there is $d'$ such that if $\hypoclass$ has $(r,s)$-threshold dimension less than $d$ whenever $r + \delta \leq s$, then $\hypoclass$ has $\gamma$-threshold dimension less than $d'$.
\end{lemma}
\begin{proof}
    It is straightforward to see that if $\hypoclass$ has $(r,s)$-threshold dimension at least $d$, then the same $a_1,\dots,a_d \in \rangespace$, $\hypo_1,\dots,\hypo_d \in \hypoclass$ that recognize this show that $\hypoclass$ has $(s - r)$-threshold dimension at least $d$.

    Now fix $0 < \delta < \gamma$ and $d$. Let $n = \lceil\frac{1}{\gamma - \delta}\rceil$. Then by Ramsey's theorem, there is some $d'$ such that any coloring of the complete graph on $d'$ vertices with $2n$ colors admits a monochromatic set of size $d$.

    We now show that if $\hypoclass$ has $\gamma$-threshold dimension at least $d'$, then there are some $r < s$ with $r + \delta \leq s$ such that $\hypoclass$ has $(r,s)$-threshold dimension at least $d$.

    By our assumption, there are $a_1,\dots,a_{d'} \in \rangespace$, $\hypo_1,\dots,\hypo_{d'} \in \hypoclass$ such that for all $i < j$, $|\hypo_j(a_i) - \hypo_i(a_j)|\geq \gamma$.

    Thus for each $i < j$, because $\frac{1}{n} \leq \gamma - \delta$, there is some $0 \leq k < n$ such that either $\hypo_j(a_i)\leq \frac{k}{n}, - \hypo_i(a_j)\geq \frac{k}{n} + \delta$, or $\hypo_i(a_j)\leq \frac{k}{n}, - \hypo_j(a_i)\geq \frac{k}{n} + \delta$.
    This gives $2n$ possible cases, so we may color the complete graph on $d'$ with $2n$ colors such that every edge $(i,j)$ with $i < j$ given a particular color falls into the same case. We may thus find a monochromatic subset $I \subseteq \{1,\dots,d'\}$ of size $d$. 
    Let $k$ be the $k$ appearing in the case corresponding to the color of that monochromatic subset, and then let $r = \frac{k}{n}, s = \frac{k}{n} + \delta$.
    Then either for all $i < j$ in $I$, $\hypo_j(a_i)\leq r, - \hypo_i(a_j)\geq s$, in which case we are done, or for all $i < j$ in $I$, $\hypo_i(a_j)\leq r, - \hypo_j(a_i)\geq s$, in which case we only need to reverse the order of our witnesses.
\end{proof}

Thus if we universally quantify over all the parameters, these two notions are the same:
\begin{corollary}
    If $\classfamily$ is a family of hypothesis classes on $\rangespace$ parameterized by $\paramspace$, then the following are equivalent:
    \begin{itemize}
        \item For every $\gamma > 0$, $\classfamily$ is uniformly $\gamma$-stable.
        \item For every $r < s$,  $\classfamily$ is uniformly $(r,s)$-stable.
    \end{itemize}
\end{corollary}
This lets us define two equivalent properties of a class $\classfamily$ - if either holds, we call $\classfamily$ \emph{uniformly stable}.

The main result in this subsection shows that the connection between stability and online learnability extends to real-valued classes:

\begin{theorem} \label{thm:fatshatteringandstable}
    A hypothesis class $\hypoclass$ is stable if and only if for every $\gamma > 0$, the sequential fat-shattering dimension $\seqfatshatter_\gamma(\hypoclass)$ is finite, and a family $\classfamily$ of hypothesis classes is uniformly stable if and only if for every $\gamma > 0$, there is some $d$ such that for each $\hypoclass \in \classfamily,$ the sequential fat-shattering dimension $\seqfatshatter_\gamma(\hypoclass)$ is at most $d$.

    Specifically, the following implications hold:
    \begin{itemize}
        \item If $\hypoclass$ has $\gamma$ sequential fat-shattering dimension less than $d$, and $r + \gamma \leq s$, then $\hypoclass$ has $(r,s)$-threshold dimension less than $2^{d+1}-1$.
        \item If $0 <\delta < \frac{\gamma}{2}$, then for every $d$ there is some $d'$ such that if $\hypoclass$ has $\delta$-threshold dimension less than $d$, then $\hypoclass$ has $\gamma$ sequential fat-shattering dimension less than $d'.$ 
    \end{itemize}
\end{theorem}

The interest in the above theorem is that stability has already been shown to be preserved in moving to from a measurable family to its expectation class:
\begin{fact}[\cite{itaykeisler,ibyrandvar}] \label{fact:preservestablerandom}
If a measurable family $\classfamily$ of hypothesis classes on $\rangespace$ parameterized by $\paramspace$ is uniformly stable, then
the expectation class $\expclass\classfamily$ is also stable.
\end{fact}

The result above is phrased very differently in the cited source, in terms of continuous logic. In Appendix \ref{app:translatestability} we explain how to get from the statement in the sources to this version.

Thus we can combine Theorem \ref{thm:fatshatteringandstable}, Fact \ref{fact:preservestablerandom}, and the characterization of agnostic online learnability via sequential fat-shattering in Fact \ref{fact:sfs_online} to give another proof that preservation of agnostic online learning in moving to statistical classes:
\begin{corollary}\label{cor:online_dist_alternate}
    For any measurable family $\classfamily$ of hypothesis classes on $\rangespace$ parameterized by $\paramspace$, the expectation class $\expclass \classfamily$ is agnostic online learnable if and only if for every $\gamma > 0$, there is some $d$ such that for each $\hypoclass$ in the range of $\classfamily$, $\hypoclass$ has $\gamma$ sequential fat-shattering dimension at most $d$.
    
    In particular, if a hypothesis class $\hypoclass$ is agnostic online learnable, so are its distribution class and dual distribution class.
\end{corollary}

We now turn to proving Theorem \ref{thm:fatshatteringandstable}:
\begin{proof}
    Assume that $\hypoclass$ has $(r,s)$-threshold dimension at least $2^{d+1}-1$, with $\gamma = s - r$. We will show that $\hypoclass$ $\gamma$ fat-shatters a  binary tree of depth $d$.
    First, we linearly order the set $\{-1,1\}^{\leq d}$ in a variation of
    lexicographical fashion, so that for any string $t_1$ of length $k$ with $k < d$, and any string $t_2$ strictly extending $t_1$, if $(t_2)_k = -1$ then $t_1 < t_2$, and if $(t_2)_k = 1$ then $t_2 > t_1$.
    By the assumption that $\hypoclass$ has $(r,s)$-threshold dimension at least $2^{d+1}-1$, as the set $\{-1,1\}^{\leq d}$ has size $2^{d+1}-1$, there are $((a_t, \hypo_t) : t \in \{-1,1\}^{\leq d})$ such that for all $i < j$ in this order,
    $\hypo_j(a_i) \leq r$ and $\hypo_i(a_j) \geq s.$

    We claim we can shatter the tree $T$ defined by sending each sequence $E \in \{-1,1\}^{< d}$ to $a_E$.
    For every $E \in \{-1,1\}^d$, let $E_{<t} = (E(0),\dots,E(t - 1))$. We consider $\hypo_E$, and see that for all $0 \leq t < d$, if $E(t) = -1$, then $E < E_{<t}$, so $\hypo_E(a_{E_{<t}}) \leq r$, while if $E(t) = 1$, we have $E > E_{<t}$ and $\hypo_E(a_{E_{<t}}) \geq s$. As $s - r \geq \gamma$, this tree is $\gamma$-shattered.
    Thus we have finished the proof of the first dimension implication in Theorem \ref{thm:fatshatteringandstable}.

    To prove the other direction, we fix $0 <\delta < \frac{\gamma}{2}$.
    We will show by induction on $d$ that if $k > (\gamma - 2\delta)^{-1}$ and $\hypoclass$ $\gamma$-shatters a binary tree $T$ in $X$ of depth $\frac{k^{d + 1} - 1}{k - 1}$, then there are $x_1,\dots,x_d \in \rangespace$ and $h_1,\dots,h_d \in \hypoclass$ such that for all $i < j$, $|h_i(x_j) - h_j(x_i)| \geq \delta$.

    The base case is $d = 1$. We just need that $X$ and $\hypoclass$ are nonempty, which is satisfied by the existence of any shattered tree.

    For the induction step, we will need a Ramsey-theoretic fact about partitions on the nodes of a binary tree, and to state it, we need to define a \emph{subtree}:
    \begin{definition}[Subtree]\label{defn:subtree}
        Let $T : \{-1,1\}^{<d} \to X$ be a binary tree in $X$ of depth $d$.
        If $x,y$ are both in the image of $T$, say that $y$ is a \emph{left descendant} of $x$ when $x = T(E)$ and $y = T(E')$, with $E'$ extending $(E_0,\dots,E_t,-1)$. If instead $E'$ extends $(E_0,\dots,E_t,1)$, we call $y$ a \emph{right descendant} of $x$.

        Then a \emph{subtree} of $T$ of depth $d' \leq d$ is a map $T' : \{-1,1\}^{<d'} \to X$ such that if $y$ is a left/right descendant of $x$ in the image of $T'$, it is also a left/right descendant of $x$ in the image of $T$.
    \end{definition}

The Ramsey-theoretic fact is:
    \begin{fact}[{\citet[Lemma 16]{private_pac}}]
        If $p,q$ are positive integers, $T : \{-1,1\}^{< p + q - 1} \to X$ is a binary tree in $X$ of depth $p + q - 1$, and the elements of $X$ are partitioned into two sets,  blue and red, then either there is a blue subtree of $T$ with depth $p$, or a red subtree of depth $q$.
    \end{fact}
    By applying this repeatedly, we find a more useful version for our purposes:
    \begin{corollary}[Tree Ramsey]\label{cor:tree_ramsey}
        If $d_1,\dots,d_k$ are positive integers, $T : \{-1,1\}^{< d_1 + \dots + d_n - k + 1} \to X$ is a binary tree in $X$ of depth $d_1 + \dots + d_n - k + 1$, and the elements of $X$ are partitioned into $k$ sets $X_1,\dots,X_k$, then there is some $i$ such that the set $X_i$ contains a subtree of $T$ of depth $d_i$.
    \end{corollary}
 
    Returning to the inductive step, assume that $d$ is such that we have the inductive invariant for $d$: if $\hypoclass$ is a class on $X$ that $\gamma$-shatters a binary tree of depth $\frac{k^{d + 1} - 1}{k - 1}$, then there are $x_1,\dots,x_d \in X$ and $h_1,\dots,h_d \in \hypoclass$ such that for all $i < j$, $|h_i(x_j) - h_j(x_i)| \geq \delta$.
    
    Now assume the hypothesis of the invariant for $d+1$: $\hypoclass$ is a class on $X$ that $\gamma$-shatters a binary tree $T$ of depth $\frac{k^{d + 2} - 1}{k - 1}$.
    In the text below, by the width of a real interval with endpoints $a<b$, we mean $b-a$.
    We pick some $h \in \hypoclass$, partition $[0,1]$ into $k$ intervals $I_1,\dots,I_k$ each of width at most $\gamma - 2\delta$, and then partition $X$ into sets $X_1,\dots,X_k$ where if $x \in X_i$, then $h(x) \in I_i$.
    Then by Corollary \ref{cor:tree_ramsey}, as $k\left(\frac{k^{d + 1} - 1}{k - 1} + 1\right) - k + 1 = \frac{k^{d + 2} - 1}{k - 1}$, the depth of $T$, some $X_{a}$ contains the set of values decorating a subtree $T'$ of $T$ of depth $\frac{k^{d + 1} - 1}{k - 1} + 1$. Let $x'$ be the root of $T'$.
    By the shattering hypothesis, there are $r$ and $s$ with $r + \gamma \leq s$ such that the tree of left descendants of $x'$ in $T'$ is $\gamma$-shattered by the set of $h'$ with $h'(x') \leq r$, and the tree of right descendants is $\gamma$-shattered by all $h'$ with $h'(x') \geq s$.
    Because $I_a$ has width at most $\gamma - 2\delta$, either the interval $[0,r]$ or the interval $[s,1]$ has distance to $I_a$ at least $\delta$. Assume without loss of generality that it is $[0,r]$.
    The tree of left descendants of $x'$ in $T'$ has depth $\frac{k^{d + 1} - 1}{k - 1}$, and is $\gamma$-shattered by the set of $h'$ with $h'(x') \leq r$.
    Thus by the inductive hypothesis, there are left descendants $x_1,\dots,x_d$ of $x'$ in $T'$ and $h_1,\dots,h_d \in \hypoclass$ with $h_i(x') \leq r$ for each $i$ such that for each $i < j \leq d$, $|h_i(x_j) - h_j(x_i)| \geq \delta$.
    We now let $x_{d + 1} = x'$ and $h_{d + 1} = h$, and observe that for $i \leq d$, $h_i(x') \leq r$ while $h(x_i) \in I_a$, so $|h_i(x') - h(x_i)| \geq \delta$.

    Thus we have completed the proof of the other direction of Theorem \ref{thm:fatshatteringandstable}.
\end{proof}

.

\subsection{Realizable Online Learning for statistical classes}
We now turn to preservation of realizable online learning for statistical classes.
As with realizable PAC learning, our results will be negative.

We start by reviewing the relationship between realizable and agnostic online learning.
Recall that for PAC learning, agnostic learnability is weaker than realizable learnability, and strictly weaker for real-valued function classes: realizable learning is a special case where the optimal hypothesis gives zero error. For realizable learning, the situation is different, since we have a stronger hypothesis on the target concept, but also a stronger requirement for our learning algorithm:  a uniform bound on regret. As with PAC learnability, there is no difference in the boundary line for learnability between realizable and agnostic for concept classes. For real-valued classes, there is a difference between agnostic and realizable learning, just as in the PAC case. 

In terms of dimensions, while agnostic online learning is characterized via the sequential fat-shattering dimension mentioned previously 
 realizable online learning has recently been characterized using  \emph{online dimension}  \citep{realizableonline}.

\begin{definition} \label{def:onlinedim}[Online dimension]
  A hypothesis class $\hypoclass$ on a set $X$
    has \emph{online dimension} greater than $D$ when there is some $d$, some $X$-valued binary tree $T : \{-1,1\}^{<d} \to X$,
    a real-valued
    binary tree 
    $\realvtree: \{-1,1\}^{<d} \to [0,1]$,
    and a $\hypoclass$-labelling of each branch in such a tree, $\branchlabelling : \{-1,1\}^d \to \hypoclass$, such that
    \begin{compactitem}
        \item for every two branches $b_0, b_1 \in \{-1,1\}^d$, if the last node at which they agree is $t$, then $$|\branchlabelling(b_0)(T(t)) - \branchlabelling(b_1)(T(t))| \geq \realvtree(t)$$
        \item for every branch $b \in \{-1,1\}^d$, whose restrictions to previous levels are $t_0,\dots,t_{d-1}$, we have $\sum_{i = 0}^{d-1}\realvtree(t_i) > D.$
    \end{compactitem}
\end{definition}

Finiteness of online dimensions characterizes realizable online learnability.
\begin{fact}[{\citet[Theorem 4]{realizableonline}}]
    Let $\hypoclass$ be a hypothesis class on a set $\rangespace$. Then $\hypoclass$ has bounded regret for realizable online learning if and only if its online dimension is finite. If $D$ is greater than the online dimension of $\hypoclass$, there is an algorithm for realizable online learning with regret at most $D$.
\end{fact}

Online dimension has a (one way) relationship to fat-shattering of binary trees:
\begin{lemma}\label{lem:online_dim_tree}
    Let $\hypoclass$ be a hypothesis class on a set $X$, let $\gamma > 0$, and $d \in \mathbb{N}$.
    If $\hypoclass$ $\gamma$-fat-shatters a tree of depth $d$, then the online dimension of $\hypoclass$ is at least $\gamma \cdot d$.
\end{lemma}

\begin{proof}
    Suppose $T$ is the $\gamma$-fat-shattered tree.
    Then fix a binary tree $s$ of depth $d$ in $\reals$, and label each branch $b \in \{-1,1\}^d$ with some $h_b \in \hypoclass$ such that for all nodes $t$ of the tree $\{-1,1\}^d$, $b_{-1},b_1$ are branches extending $t$, and $b_i$ extends $t$ concatenated with $i$ for $i = \pm 1$, then
    $$h_{b_{-1}}(T(t)) \leq s(t) - \frac{\gamma}{2},$$
    while
    $$h_{b_{1}}(T(t)) \geq s(t) + \frac{\gamma}{2}.$$

    We now show that for any $\epsilon > 0$, the online dimension of $\hypoclass$ is greater than $d(\gamma - \epsilon)$.
    We let $\realvtree : \{-1,1\}^{<d} \to [0,1]$ be the real-valued labelled binary tree with constant value $\gamma - \epsilon$, and let $\branchlabelling$ label each branch $b$ with $h_b$.
    Then for any two branches, we may without loss of generality call the branches $b_{-1}, b_1$, let $t$ be the last node at which they agree, and assume that $b_i$ extends $t$ concatenated with $i$ for $i = \pm 1$. Then
    \begin{align*}
   |h_{b_{1}}(T(t)) - h_{b_{-1}}(T(t))|
    \geq \\
    \left|\left(s(t) + \frac{\gamma}{2}\right) - \left(s(t) - \frac{\gamma}{2}\right)\right| = \gamma > \gamma - \epsilon.
     \end{align*}

    Thus $T,\realvtree,\branchlabelling$ satisfy the requirements to show that the online dimension of $\mathcal{H}$ is greater than
    $$\min_{b \in \{-1,1\}^d}\sum_{t = 0}^{d-1}\realvtree(t_i),$$
    where $t_i$ is the restriction of $b$ to level $i$.
    As $\realvtree$ takes a constant value $\gamma - \epsilon$, the online dimension is greater than $d(\gamma - \epsilon)$.
\end{proof}

From the lemma we infer an important consequence, saying that the containment between realizable and agnostic goes the opposite way in online learning as compared to PAC learning:

\begin{corollary} If $\hypoclass$ is realizable online learnable, it has some finite online dimension $D$, and thus for any $\gamma > 0$, $\hypoclass$ has sequential $\gamma$-fat-shattering dimension at most $\frac{D}{\gamma}$. 

Because this gives a finite bound for all $\gamma$, we see that:

\emph{realizable online learnability implies agnostic online learnability, even for real-valued function classes}. 
\end{corollary}

We are now ready to show that in the case of real-valued functions,  moving from a base class to statistical classes  does not preserve realizable online learnability. In fact, this will follow from lack of closure under dualization:

\begin{proposition} \label{prop:nonclosurerealizableonline} There is
a real-valued hypothesis class $\hypoclass$ that is realizable online
learnable, but its dual class is not realizable online learnable. Thus
the dual distribution class based on $\hypoclass$ is not online learnable.
\end{proposition}

Before proving the proposition, we note a pathology for realizable online learning which will introduce the example classes that are relevant to the proof of the proposition. We show that the notion of learnability is less robust, in the sense that it is not preserved under composition with increasing homeomorphisms of $[0,1]$.
\begin{theorem} \label{thm:realizableandagnosticdifferonlineregression}
    There is a hypothesis class $\hypoclass$ on a set $X$ that has finite $\gamma$ sequential fat-shattering dimension for all $\gamma > 0$, but has infinite online dimension.
    In fact, there are classes $\hypoclass, \hypoclass'$ on the same set $X$, both indexed by a set $Y$, such that both have finite $\gamma$ sequential fat-shattering dimension for all $\gamma > 0$,
    $\hypoclass$ has infinite online dimension, $\hypoclass'$ has finite online dimension, and there is an increasing homeomorphism $f : [0,1] \to [0,1]$ such that $f \circ \hypoclass = \hypoclass'$, as functions $X \times Y \to [0,1]$.

    In terms of learning, this means that both classes are agnostic online learnable, but only $\hypoclass'$ is realizable online learnable. In particular, this show that the dividing lines for these notions of learnability are different.
\end{theorem}

We proceed to the proof of the theorem.
\begin{proof}
    Let $X$ be the infinite binary tree $X = \bigcup_{t = 0}^\infty\{0,1\}^t$.
     Fix a decreasing sequence $\Gamma= \gamma_0,\gamma_1,\dots$ of positive reals to be determined, with $\lim_d \gamma_d = 0$. We define $\hypoclass^\Gamma$, a hypothesis class indexed by the  infinite branches $\{0,1\}^{\mathbb{N}}$ of that infinite binary tree.
   
    Given an infinite branch $b$, and $x \in X$, the hypothesis $h_b(x) = 0$ when $x$ is not an initial segment of $b$. If $x$ is an initial segment of $b$, let $d$ be its length. 
    Then we let $h_b(x) = b_d \cdot \gamma_d$.
    Thus for any $x \in \{0,1\}^d$ and branches $b, b'$, the difference $|h_b(x) - h_{b'}(x)|$ is either 0 or $\gamma_d$, with the latter only occurring when at least one of $b, b'$ extends $x$. In particular, if $b_1,b_2,b_3 \in \{0,1\}^\N$ are pairwise distinct, then there must be some pair $i \neq j$ with $i,j \in \{1,2,3\}$ with $|h_{b_i}(x) - h_{b_j}(x)| = 0$, as otherwise, we must have $(b_i)_d \neq (b_j)_d$ for each $i \neq j$, but $(b_1)_d,(b_2)_d,(b_3)_d \in \{0,1\}$, so all three bits cannot be pairwise distinct.

    For any $\gamma > 0$, we will calculate that $\hypoclass^\Gamma$ has finite $\gamma$ sequential fat-shattering dimension.
    Specifically, if $\hypoclass^\Gamma$ $\gamma$ fat-shatters a binary tree, it is clear that the nodes of this tree must all be nodes of $X$ of length at most $d$, where $d$ is the largest number such that $\gamma_d \geq \gamma$.
    Thus the depth of the $\gamma$ fat-shattered tree must be at most $d$, so the $\gamma$ sequential fat-shattering dimension is at most $d$. In particular, regardless of the rate at which  $\gamma_i$ goes to zero, the resulting class is agnostic online learnable.

    We now characterize when the class is realizable online learnable, which will depend on the rate at which the parameters $\gamma_i$ go to 0.
    
    Note that the tree $\{0,1\}^{<d} \subseteq X$ is $\gamma_d$ fat-shattered: for each maximal branch $b$ of the tree, we label the branch with a hypothesis corresponding to any infinite branch extending $b$. Thus by Lemma \ref{lem:online_dim_tree}, $\hypoclass^\Gamma$ will have online dimension at least $d \cdot \gamma_d$. If the sequence $(d \cdot \gamma_d : d \in \mathbb{N})$ is unbounded, then the online dimension is infinite, and there is no uniform bound for regret for realizable online learning.

    Now we will show that if $\Gamma$ is chosen such that the sum $\sum_{i = 0}^\infty 2^i\gamma_i$ converges, then the online dimension of $\hypoclass^\Gamma$ is at most $\sum_{i = 0}^\infty 2^i \gamma_i$, and in particular, $\hypoclass^\Gamma$ has finite online dimension. Assume for contradiction that the online dimension is greater than $\sum_{i = 0}^\infty 2^i \gamma_i$. For this to be true, it must be witnessed by some $d$, an $X$-valued tree $T$ of depth $d$, a tree $\realvtree$ of real-valued errors, and an assignment of elements from $\hypoclass$ to each branch of the tree.

    We claim that if $s, t \in \{-1,1\}^{<d}$ are such that $s$ is a strict initial substring of $t$ and $T(s) = T(t)$, then either $\realvtree(s) = 0$ or $\realvtree(t) = 0$.
    Let $b_1,b_2 \in \{-1,1\}^d$ be two branches extending $t$, while $b_3$ extends $s$ in such a way that its last common node with $b_1,b_2$ is $s$.
    As noted earlier, there must be two of these three branches such that $i \neq j$ but $\branchlabelling(b_i)(T(s)) = \branchlabelling(b_j)(T(s))$.
    If $\branchlabelling(b_1)(T(s)) = \branchlabelling(b_2)(T(s))$, then as $T(s) = T(t)$,
    $$\realvtree(t) \leq |\branchlabelling(b_1)(T(t)) - \branchlabelling(b_2)(T(t))| = 0.$$
    Otherwise, for some $i \in \{1,2\}$, we have $\branchlabelling(b_i)(T(s)) = \branchlabelling(b_3)(T(s))$,
    so
    $$\realvtree(s) \leq |\branchlabelling(b_i)(T(s)) - \branchlabelling(b_3)(T(s))| = 0.$$

    Now consider a branch of $\{-1,1\}^{<d}$ consisting of nodes $t_0,\dots,t_{d-1}$. We can bound the sum of the weights of that branch by grouping the indices $i$ by the value of $T(t_i)$:
    $$\sum_{i = 0}^{d - 1}\realvtree(t_i)
    \leq \sum_{x \in X}\left(\sum_{0 \leq i < d - 1 :\, T(t_i) = x}\realvtree(t_i)\right).
    $$
    For each $x \in X$, there is at most one $i$ such that $T(t_i) = x$ and $\realvtree(t_i) > 0$, so only one $i$ can contribute to the sum
    $\sum_{x \in X}\left(\sum_{0 \leq i < d - 1 :\, T(t_i) = x}\realvtree(t_i)\right)$.
    If $x$ has length $\ell$ and $i$ is such that $T(t_i) = x$,
    then $\realvtree(t_i) \leq \gamma_\ell$, so 
    $$\sum_{x \in X}\left(\sum_{0 \leq i < d - 1 :\, T(t_i) = x}\realvtree(t_i)\right)\leq \gamma_\ell.$$
    As there are only $2^\ell$ elements of $X$ with length $\ell$,
    the online dimension $D$ of $\hypoclass$ is bounded by
    $$D < \sum_{i = 0}^{d - 1}\realvtree(t_i)
    \leq \sum_{x \in X}\left(\sum_{0 \leq i < d - 1 :\, T(t_i) = x}\realvtree(t_i)\right) \leq \sum_{\ell=0}^\infty 2^\ell\gamma_\ell.$$
    By assumption, this latter sum converges, so the online dimension is finite.

    We now claim that if $\hypoclass^\Gamma$ is constructed from the sequence $\Gamma=\gamma_1,\gamma_2,\dots$ while $\hypoclass^{\Gamma'}$ is constructed in the same way from the sequence $\Gamma'= \gamma'_1,\gamma'_2,\dots$, then there is an increasing homeomorphism $f : [0,1] \to [0,1]$ such that $f \circ \hypoclass^\Gamma = \hypoclass^{\Gamma'}$.
    To do this, define $f(1) = 1$, and for each $d$, define $f(\gamma_d) = \gamma'_d$
    We can then extend this to a piecewise linear definition, with countably many pieces, on $(0,1]$, where $\lim_{x \to 0}f(x) = 0$, so defining $f(0) = 0$ will maintain continuity.

    We now see that by choosing $\Gamma= \gamma_1,\gamma_2,\dots$ so that $\lim_d d \cdot \gamma_d = \infty$ and choosing $\Gamma'= \gamma'_1,\gamma'_2,\dots$ so that $\sum_{i = 1}^\infty \gamma'_i$ converges, we find $\hypoclass^\Gamma$ with infinite online dimension and $\hypoclass^{\Gamma'}$ with finite online dimension such that $f \circ \hypoclass^\Gamma = \hypoclass^{\Gamma'}$.
\end{proof}

Using the same family of examples, we now prove Proposition \ref{prop:nonclosurerealizableonline}:

\begin{proof}
Let $\Gamma= \gamma_i: i>0$ be a sequence such that the  $i \cdot \gamma_i$ is unbounded.
Let $\hypoclass^\Gamma$ be the class from Theorem \ref{thm:realizableandagnosticdifferonlineregression}. Recall that range points are prefixes $p$ (finite sequences). Hypotheses are parameterized by infinite sequences $s$ and the value of a hypothesis $h_s$ on a prefix $p$ is either zero or $\frac{1}{\gamma_n}$ for $n$ the length of $p$.  Then, as proven in Theorem \ref{thm:realizableandagnosticdifferonlineregression} $H_\gamma$ is not online learnable in the realizable case.

Let $D_\gamma$ be the dual class. So points are now $\omega$-sequences $s$, hypotheses are prefixes $p$, and the value of a hypothesis $h_p$  at a sequence $s$ is $0$ if $s$ does not extend $p$ and is $\frac{1}{\gamma_n}$ if
$s$ does extend $p$, where again $n$ is the length of prefix $p$. We claim that $D_\gamma$ \emph{is} realizable online learnable.
Consider the definition of online learnability in terms of a game between learner and adversary.
Adversary is playing range points for $D_\gamma$ -- that is, infinite sequences $s$. And at a move for learner with previous adversary range points $s_1 \ldots s_k$, learner knows the values $v_1 \ldots v_{k-1}$ of a $D_\gamma$-consistent hypothesis for $s_1 \ldots s_{k-1}$.
Note that if adversary ever reveals a value $v_i$ that is non-zero, learner will know the hypothesis, since there is a unique prefix of $s_i$ that would give such a value.
Thus adversary should always reveal value zero. Thus learner has a strategy that will achieve bounded loss: play zero until a non-zero value appears.

To finish the proof, we  note that $H_\gamma$ is the dual of $D_\gamma$.
\end{proof}

\subsection{An Alternate Approach to Realizable Online Learning}
Thus far we have shown that realizable online learning is not preserved under moving to statistical classes. We have indicated that this is related to another pathology, that online learnability is not preserved under applying continuous mappings.  We will now look at using this intuition to ``fix'' the pathologies of realizable online learning. We will do this by changing the definition,
applying alternate loss functions which do not sum the cumulative losses, but rather discretize them.

Recall that a run of a learning algorithm for $T$ rounds yields a sequence $(z_1,y_1,y_1'),\dots,(z_T,y_T,y_T')$, where $(z_i,y_i)$ are the moves by the adversary, and $y_i'$ are the moves by the learner.
Earlier, the loss was defined as $\sum_{i \leq T} |y_i' - y_i|$. Here we let the loss be $\sum_{i \leq T} \ell(|y_i' - y_i|)$, for certain $\ell : [0,1] \to [0,\infty)$  nondecreasing.
We then define regret in terms of this new loss function $\ell$, and the remainder of the setup remains unchanged, including the restriction on the adversary in the realizable case. 

\begin{definition}[Online learnability for a general loss function]
    Given a loss function $\ell : [0,1] \to [0,\infty)$, say that a hypothesis class $\hypoclass$ is \emph{$\ell$-online learnable in the agnostic case} when there is a learning algorithm whose minimax regret \emph{with loss function $\ell$} against any adversary is sublinear in $T$.

    Say that a hypothesis class $\hypoclass$ is \emph{$\ell$-online learnable in the realizable case} when there is a learning algorithm whose minimax regret \emph{with loss function $\ell$} against any realizable adversary is bounded, uniform in $T$.
\end{definition}

The lower the loss function, the easier it is to learn:

\begin{lemma}\label{lem:dominated_loss}
    Let $C > 0$, and suppose $\ell_1,\ell_2 : [0,1] \to [0,1]$ are loss functions with $C\ell_1(x) \leq \ell_2(x)$ for all $x \in [0,1]$.

    Then in either the agnostic or realizable case, if a hypothesis class $\hypoclass$ is $\ell_2$-online learnable, then it is $\ell_1$-online learnable.
\end{lemma}
\begin{proof}
    The same learning algorithm will suffice.
    On any given run of the algorithm, the regret with $\ell_1$ as the loss function will be at most the regret with $\ell_2$ as the loss function:
    $$C\sum_{i \leq T}\ell_1(|y_i' - y_i|) \leq \sum_{i \leq T}\ell_2(|y_i' - y_i|),$$
    so regret with loss function $\ell_1$ will satisfy the same upper bound required of regret with loss function $\ell_2$, up to the fixed factor $C$.
\end{proof}
We now define the loss functions we will focus on:
\begin{definition}[$\epsilon$-truncated loss functions]
    Given $\epsilon > 0$, define $\ell_\epsilon, L_\epsilon : [0,1] \to [0,\infty)$ by
    \begin{align*}
        \ell_\epsilon(x) &= \max(x - \epsilon, 0)\\
        L_\epsilon(x) &= \begin{cases}
            0 & \textrm{if }x < \epsilon\\
            1 & \textrm{if }x \geq \epsilon.
        \end{cases}
    \end{align*}
\end{definition}

The usual loss function we just denote by $\idloss$ (this is the identity, so $\idloss(x) = x$).
Using these other loss functions make it easier for a class to be online learnable:

\begin{lemma}\label{lem:loss_implications}
    In either the agnostic or realizable case, online learnability implies $L_\epsilon$-online learnability which implies $\ell_\epsilon$-online learnability.
\end{lemma}

\begin{proof}
    Online learnability is equivalent to the standard notion $\idloss$-online learnability,  and we see that for any $x \in [0,1]$,
    \begin{align*}
        \epsilon L_\epsilon(x) &\leq \idloss(x)\\
        \ell_\epsilon(x) &\leq L_\epsilon(x),
    \end{align*}
    so the result follows by Lemma \ref{lem:dominated_loss}.
\end{proof}

The reason for studying these loss functions is that when we require online learnability with respect to any  $\ell_\epsilon$, the gap between agnostic and realizable online learnability vanishes. The proof below will apply prior characterizations of agnostic online learnability, which are given in terms of notions from model theory.
\begin{theorem} \label{thm:equivalternaterealizableonline}
    For a hypothesis class $\hypoclass$, the following are equivalent:
    \begin{compactitem}
        \item $\hypoclass$ is online learnable in the agnostic case
        \item  $\forall \epsilon > 0$, $\hypoclass$ is $L_\epsilon$-online learnable in the agnostic case
        \item $\forall \epsilon > 0$, $\hypoclass$ is $\ell_\epsilon$-online learnable in the agnostic case
        \item  $\forall \epsilon > 0$, $\hypoclass$ is $L_\epsilon$-online learnable in the realizable case
        \item  $\forall \epsilon > 0$, $\hypoclass$ is $\ell_\epsilon$-online learnable in the realizable case
    \end{compactitem}
\end{theorem}

As a corollary, we see that these properties of $\hypoclass$, being equivalent to online learnability in the agnostic case, are also preserved under moving to statistical classes:
\begin{corollary} \label{cor:alternativerealizablerandomization}
    If for every $\epsilon > 0$, $\hypoclass$ is $\ell_\epsilon$-online learnable in the realizable case, or for every $\epsilon > 0$, $\hypoclass$ is $L_\epsilon$-online learnable in the realizable case,  then the same holds 
    for both the distribution class and the dual distribution class.
\end{corollary}

We now work towards the proof of Theorem \ref{thm:equivalternaterealizableonline}.

To handle realizable online learnability with a loss function, we need to expand online dimension to include a loss function. In fact, \citep{realizableonline} defined it for an even more broad notion of a loss function, but this version will suffice for our purposes here:
\begin{definition} \label{def:onlinedimgen}[Online dimension for a general loss function]
  In the context of a loss function $\ell : [0,1] \to [0,1]$, a hypothesis class $\hypoclass$ on a set $X$
    has \emph{online dimension} greater than $D$ when there is some $d$, some $X$-valued binary tree $T : \{-1,1\}^{<d} \to X$,
    a real-valued
    binary tree 
    $\realvtree: \{-1,1\}^{<d} \to [0,1]$,
    and a $\hypoclass$-labelling of each branch in such a tree, $\branchlabelling : \{-1,1\}^d \to \hypoclass$, such that
    \begin{itemize}
        \item for every two branches $b_0, b_1 \in \{-1,1\}^d$, if the last node at which they agree is $t$, then $$\ell\left(|\branchlabelling(b_0)(T(t)) - \branchlabelling(b_1)(T(t))|\right) \geq \realvtree(t)$$
        \item for every branch $b \in \{-1,1\}^d$, whose restrictions to previous levels are $t_0,\dots,t_{d-1}$, we have $\sum_{i = 0}^{d-1}\realvtree(t_i) > D.$
    \end{itemize}
\end{definition}

This dimension still characterizes realizable online learnability, as in the full version of \citep[Theorem 4]{realizableonline}:
\begin{fact}[{\citet[Theorem 4]{realizableonline}}]\label{fact:online_dimension_loss}
    Let $\hypoclass$ be a hypothesis class on a set $X$. Then when $\ell : [0,1] \to [0,1]$ is a loss function, $\hypoclass$ has bounded regret for realizable online learning if and only if $\mathrm{Onl}_\ell(\hypoclass) < \infty$.

    Specifically, if $\mathrm{Onl}_\ell(\hypoclass) < D$, then there is an algorithm for realizable online learning with regret at most $D$ with respect to loss function $\ell$. Conversely, if $\mathrm{Onl}_\ell(\hypoclass) > D$, then for any algorithm for realizable online learning, the minimax regret with respect to loss function $\ell$ is at least $\frac{D}{2}$.
\end{fact}

\begin{lemma}\label{lem:shatter_loss}
For any non-decreasing loss function $\ell : [0,1] \to [0,1]$, 
    if a hypothesis class $\hypoclass$ on $X$ $\gamma$ sequentially fat-shatters a binary tree in $X$ of depth $d$, then $\mathrm{Onl}_\ell(\hypoclass) \geq d \cdot \ell(\gamma)$, and also the minimax regret of a $d$-round online learner in the agnostic case with loss function $\ell$ is at least $\frac{1}{3}  d \cdot \ell(\gamma)$.
\end{lemma}
\begin{proof}
We first show that $\mathrm{Onl}_\ell(\hypoclass) \geq d \cdot \ell(\gamma)$. This is only a slight modification of Lemma \ref{lem:online_dim_tree}.

As in that proof, suppose $T$ is the $\gamma$-fat-shattered tree. Let the binary tree $s : \{-1,1\}^{<d} \to \reals$ and the branch labelling $\Lambda$ which labels each branch $b \in \{-1,1\}$ with $h_b$ be as in that proof.
Then as in that proof, for any two branches, we may refer to those branches without loss of generality as $b_{-1},b_1$, where $t$ is the last node at which they agree, and assume that $b_i$ extends $t$ concatenated with $i$ for $i = \pm 1$.
The essential property of sequential fat-shattering is that then
$$|h_{b_1}(T(t)) - h_{b_{-1}}(T(t))|\geq \gamma,$$
so we may choose the real labelling $\realvtree : \{-1,1\}^{<d} \to [0,1]$ given by
$\tau(t) = \ell(\gamma) - \epsilon$,
and then $T,\tau,\Lambda$ satisfy the requirements to show that the online dimension of $\hypoclass$ is greater than
    $$\min_{b \in \{-1,1\}^d}\sum_{t = 0}^{d-1}\realvtree(t_i),$$
    where $t_i$ is the restriction of $b$ to level $i$.
    As $\realvtree$ takes a constant value $\ell(\gamma) - \epsilon$, the online dimension is greater than $d(\ell(\gamma) - \epsilon)$.
\end{proof}

The loss functions $L_\epsilon$ were chosen so that online dimension would capture sequential fat-shattering dimension:
\begin{lemma}\label{lem:sfs_loss}
    For any hypothesis class $\hypoclass$ and any $\epsilon > 0$, if $\mathrm{Onl}_{L_\epsilon}(\hypoclass)$ is infinite, then $\hypoclass$ is not agnostic online learnable.
\end{lemma}
\begin{proof}
Here we will use the connection of agnostic online learnability with stability, which was utilized earlier in Section \ref{subsec:onlinelearningviastability}. 

    Specifically we will show that for $0 < \delta < \frac{\gamma}{2}$, $\hypoclass$ is not $\delta$-stable, which suffices by Theorem \ref{thm:fatshatteringandstable}.

    Because $L_\epsilon$ is $\{0,1\}$-valued, if $D$ is an integer and $\mathrm{Onl}_{L_\epsilon}\geq D$, then there is a tree $T:\{-1,1\}^{<d} \to X$, a $\{0,1\}$-valued binary tree $\realvtree : \{-1,1\}^{<d} \to \{0,1\}$, and an $\hypoclass$-labelling of each branch in the tree, $\branchlabelling : \{-1,1\}^d\to\hypoclass$, such that
    \begin{itemize}
        \item for every two branches $b_0,b_1 \in \{-1,1\}^d$, if the last node at which they agree is $t$, and 
        $|\branchlabelling(b_0)(T(t)) - \branchlabelling(b_1)(T(t))|\geq\epsilon$,
        then $\realvtree(t)= 1$.
        \item for every branch $b \in \{-1,1\}^d$ whose restrictions to previous levels are $t_0,\dots,t_{d-1}$, at least $D$ of these nodes have $\tau(t_i) = 1$.
    \end{itemize}

    We will show that there must also be a binary tree of depth $D$ that is $\epsilon$ fat-shattered by $\hypoclass$.
    To do this recursively. It will helpful below for the reader recall the definitions of descendants and subtrees from Definition \ref{defn:subtree}.
    
    We claim that for any integer $D$, if every branch of a finite-depth $\{0,1\}$-valued binary tree $t : \{-1,1\}^{<d} \to \{0,1\}$ has at least $D$ nodes labelled 1, then there is a depth-$D$ subtree of this binary tree with all nodes labelled 1. That is, there is a function $\iota : \{-1,1\}^{<D} \to \{-1,1\}^{<d}$, increasing in the tree order, such that $t \circ \iota(v) = 1$ for all nodes $v$.
    
    We construct the subtree recursively, inducting on $D$. This is trivial for $D = 0$. Suppose that this is true for $D$, and we now consider a finite-depth $\{0,1\}$-valued binary tree $t : \{-1,1\}^{<d} \to \{0,1\}$ where every branch has at least $D + 1$ nodes labelled 1. Let $v$ be a node of minimal length with $\realvtree(v) = 1$. We now consider the tree of all left descendants of $v$. Suppose this tree has depth $d'$. Then the labelling $t$ on this subtree induces a labelling $t' : \{-1,1\}^{<d'} \to \{0,1\}$, given by concatenating each node $w \in \{-1,1\}^{<d'}$ with $v$ and $-1$ to form a left descendant $w'$ of $v$, and then setting $t'(w) = t(w')$. Every branch $b \in \{-1,1\}^{d'}$ of this smaller tree corresponds uniquely to a branch $b' \in \{-1,1\}^d$ of the original tree which extends $v$ to the left. To see this, let $v = (v_0,\dots,v_\ell)$, and let $b = (b_0,\dots,b_{d' - 1})$ - we then define this new branch to be $b' = (v_0,\dots,v_\ell,-1,b_0,\dots,b_{d'-1})$.
    We then see that of the nodes leading up to $b'$, at least $D+1$ are labelled $1$ by $t$. Let $S$ be the set of such nodes. Because the elements of $S$ lie on the same branch, they are comparable. As this branch contains $v$, they are thus either substrings of $v$, or descendants of $v$. By the minimality assumption, the only substring of $v$ which can be labelled 1 by $t$ is $v$, so the remaining $\geq D$ elements of $S$ are descendants of $v$. As these lie on the branch $b'$, they are left descendants. These correspond to nodes of the smaller tree $\{-1,1\}^{<d'}$, which lie on the branch $b$, and are labelled 1 by $t'$.
    Thus every branch $b'$ contains at least $D$ elements labelled 1 by $t'$, so the induction hypothesis applies. There is thus a subtree of $\{-1,1\}^{<d'}$ of depth $D$ where all nodes are labelled $1$ by $t'$ - this is given by a map $\iota_{-1} : \{-1,1\}^{<D} \to \{-1,1\}^{<d'}$ such that for all nodes $w \in \{-1,1\}^{<D}$, concatenating $v$ with $-1$ and $\iota_{-1}(w)$ gives a node $w'$ with $t(w') = 1$. The same must be true, by symmetry, of the right descendants, and these two trees, together with $v$, form a subtree of depth $D + 1$, with all nodes labelled 1.

    We now return to our trees $T,\realvtree$, and our branch labelling $\Lambda$. By our claim, there is a subtree of $\{-1,1\}^{<d}$ of depth $D$ where every element is labelled $1$ by $\realvtree$. We can thus choose an increasing (in the tree partial order) function $\iota : \{-1,1\}^{<D} \to \{-1,1\}^{<d}$ with $\realvtree \circ \iota(v) = 1$ for all $v$. 
    For every branch $b \in \{1,-1\}^D$ of $\{1,-1\}^{<D}$, choose a branch $b' \in \{-1,1\}^d$ extending $\iota(v)$ for all nodes $v$ comprising $b'$. Label $b$ with $\branchlabelling'(b)  = \branchlabelling(b')$.

    Now for any two branches $b_{-1},b_1$ of $\{-1,1\}^D$, if $t$ is the last node at which they agree, we can assume that $b_i$ extends $t$ concatenated with $i$ for $i = \pm 1$. As above, for $i = \pm 1$, let $b'_i \in \{-1,1\}^d$ be the chosen branch corresponding to $b_i$ in the larger tree. Because we have only chosen nodes labelled with $1$, we have
    \begin{align*}
        L_\epsilon(|\branchlabelling'(b_{-1})(T \circ \iota(t)) - \branchlabelling'(b_1)(T \circ \iota(t))|)
        &= L_\epsilon(|\branchlabelling(b_{-1}')(T \circ \iota(t)) - \branchlabelling(b_1')(T \circ \iota(t))|)\\
        &\geq \realvtree \circ \iota(t)\\
        &= 1,
    \end{align*}
    so in particular, 
    $|\branchlabelling'(b_{-1})(T \circ \iota(t)) - \branchlabelling'(b_1)(T \circ \iota(t))| \geq \epsilon$.

We will connect this to stability using an argument that is a slight modification of the proof of one direction of Theorem \ref{thm:fatshatteringandstable}. Fix $0 < \delta < \frac{\epsilon}{2}$ and $k > (\epsilon - 2\delta)^{-1}$. We wish to show that for all $d$, there are $x_1,\dots,x_d \in X$ and $h_1,\dots,h_d \in \hypoclass$ such that for all $i < j$, $|h_i(x_j) - h_j(x_i)| \geq \delta$.

First, we observe that in the first part of this proof, we have shown that for all $d$, there exists an $X$-valued binary tree $T$ of depth $\frac{k^{d + 1} - 1}{k - 1}$,
and a labelling $\branchlabelling$ of the branches of $T$ with elements of $\hypoclass$ such that for any branches $b_{-1},b_1$, if $t$ is the last node at which they agree, then
$|\branchlabelling(b_{-1})(T(t)) - \branchlabelling(b_1)(T(t))| \geq \epsilon$. We say that such a tree $T$ is $\epsilon$ \emph{spread-shattered} by $\hypoclass$, and that $\branchlabelling$ \emph{witnesses spread-shattering} of $T$.

To finish the proof, we show by induction on $d$ that if $\hypoclass'$ is a class on $X$ that $\epsilon$ spread-shatters an $X$-valued binary tree $T$ of depth $\frac{k^{d + 1} - 1}{k - 1}$, then there are $x_1,\dots,x_d \in X$ and $h_1,\dots,h_d \in \hypoclass'$ such that for all $i < j$, $|h_i(x_j) - h_j(x_i)| \geq \delta$. As this holds for $\hypoclass$ and any $d$, the result follows.

Consider the base case,  $d = 1$. We simply need that $X$ and $\hypoclass'$ are nonempty, which is satisfied by the existence of any shattered tree.

Now assume that $d$ is such that we have the inductive invariant for $d$: if $\hypoclass'$ is a class on $X$ that $\epsilon$ spreads-shatters a binary tree of depth $\frac{k^{d + 1} - 1}{k - 1}$, then there are $x_1,\dots,x_d \in X$ and $h_1,\dots,h_d \in \hypoclass'$ such that for all $i < j$, $|h_i(x_j) - h_j(x_i)| \geq \delta$.
    
Also assume the hypothesis of the invariant for $d+1$: $\hypoclass'$ is a class on $X$ that $\epsilon$ spread-shatters a binary tree $T$ of depth $\frac{k^{d + 2} - 1}{k - 1}$.
As before, by the width of a real interval with endpoints $a<b$, we mean $b-a$.
We pick some $h \in \hypoclass'$, partition $[0,1]$ into $k$ intervals $I_1,\dots,I_k$ each of width at most $\epsilon - 2\delta$, and then partition $X$ into sets $X_1,\dots,X_k$ where if $x \in X_i$, then $h(x) \in I_i$.
    Then by our Ramsey result for trees, Corollary \ref{cor:tree_ramsey}, as $k\left(\frac{k^{d + 1} - 1}{k - 1} + 1\right) - k + 1 = \frac{k^{d + 2} - 1}{k - 1}$, the depth of $T$, some $X_{a}$ contains the set of values decorating a subtree $T'$ of $T$ of depth $\frac{k^{d + 1} - 1}{k - 1} + 1$. Let $x'$ be the root of $T'$.

    Let $\hypoclass_L \subseteq \hypoclass'$ consist of all hypotheses $\branchlabelling(b)$ where $b$ is a branch of $T$ extending $x'$ to the left, and let $\hypoclass_R \subseteq \hypoclass'$ consist of all hypotheses $\branchlabelling(b)$ where $b$ is a branch of $T$ extending $x'$ to the right. By the spread-shattering hypothesis, if $h_L \in \hypoclass_L$ and $h_R \in \hypoclass_R$, then $|h_L(x') - h_R(x')|\geq \epsilon$. Because $I_a$ has width at most $\epsilon - 2\delta$, either the set $\{h_L(x'):h_L \in \hypoclass_L\}$ or the set $\{h_L(x'):h_L \in \hypoclass_L\}$ has distance to $I_a$ at least $\delta$. Assume without loss of generality that it is $\{h_L(x'):h_L \in \hypoclass_L\}$.
    The tree of left descendants of $x'$ in $T'$ has depth $\frac{k^{d + 1} - 1}{k - 1}$, and is $\epsilon$-shattered by $\hypoclass_L$.
    Thus by the inductive hypothesis, there are left descendants $x_1,\dots,x_d$ of $x'$ in $T'$ and $h_1,\dots,h_d \in \hypoclass_L$ for each $i$ such that for each $i < j \leq d$, $|h_i(x_j) - h_j(x_i)| \geq \delta$.
    We now let $x_{d + 1} = x'$ and $h_{d + 1} = h$, and observe that for $i \leq d$, $h_i \in \hypoclass$ while $h(x_i) \in I_a$, so $|h_i(x') - h(x_i)| \geq \delta$.
\end{proof}

We are now ready to prove Theorem \ref{thm:equivalternaterealizableonline}:

\begin{proof}
    Many of these implications follow from Lemma \ref{lem:loss_implications}.
    To show all of the agnostic case conditions are equivalent, it suffices to show that if for every $\epsilon > 0$, $\hypoclass$ is $\ell_\epsilon$-online learnable, $\hypoclass$ is online learnable.

    We show this through the contrapositive. If $\hypoclass$ is not online learnable in the agnostic case, then there is some $\epsilon > 0$ such that $\hypoclass$ sequentially $2\epsilon$ fat-shatters a tree of depth $d$ for every $d$. By Lemma \ref{lem:shatter_loss}, the regret of agnostic $\ell_\epsilon$-online learning in $d$ rounds is at least $d\ell_\epsilon(2\epsilon) = d\epsilon$, so this is not sublinear.

    To show that the realizable case conditions are equivalent to agnostic online learnability, we first observe that by Lemma \ref{lem:sfs_loss}, if $\hypoclass$ is online learnable in the agnostic case, then for any $\epsilon > 0$, $\hypoclass$ is $L_\epsilon$-online learnable in the realizable case.
    
    Using Lemma \ref{lem:loss_implications} once more, it suffices to show that if for every $\epsilon > 0$, $\hypoclass$ is $\ell_\epsilon$-online learnable in the realizable case, $\hypoclass$ is online learnable in the agnostic case. The contrapositive of this follows from the other part of Lemma \ref{lem:shatter_loss}.
\end{proof}

\section{Related Work} \label{sec:related}
 As noted earlier, what we refer to as the ``dual distribution class'' is from \citep{sigmod22}, where it is defined only for concept classes. The basic result of \citep{sigmod22} is that agnostic PAC learnability of a base concept class implies agnostic PAC learnability of the dual distribution class, with accompanying sample complexity bounds that are asymptotically looser than ours. 
 The extension of dual distribution classes to base classes consisting of real-valued values, as well as  the notion of the distribution class, is new to our work.

We have presented our arguments without relying on logical notions. But our work is inspired by work on learnability of
hypothesis classes coming from logical formulas, so we say more about the connection here.  Given a structure $\amodel$ and a logical formula $\phi(x_1 \ldots x_k)$,  one can look at the $k$-tuples from the domain of $\amodel$ that satisfy $\phi$: this is a \emph{definable set of tuples} from $\amodel$.  Given
a first-order formula $\phi(x_1 \ldots x_j;y_1 \ldots y_k)$ in which the free variables are partitioned, we get a family of sets of $j$-tuples, by varying $\vec y$ over all $k$-tuples in $\amodel$. Thus each $\phi(\vec x;\vec y)$ gives a concept class with range space the $j$-tuples from the domain of $\amodel$, and parameter space the $k$-tuples from the structure. For example the family of rectangles in the reals and the family of intervals in the reals are definable families. A partitioned formula is said to be \emph{NIP} if the corresponding  concept class has finite VC dimension. 

Model theorists have identified a number of structures where the concept class arising from each partitioned formula has finite VC dimension (or equivalently, is agnostic PAC learnable):  these are called \emph{NIP structures}: see, for example, \citep{nipbook}. As discussed in our paper, agnostic online learnability corresponds to a hypothesis class being stable: see Definition \ref{def:stableconcept}. A partitioned formula $\phi(\vec x;\vec y)$ in a structure $\amodel$ is called stable exactly when the corresponding concept class is stable in the sense of Definition \ref{def:stableconcept}. The study of stable formulas has been developed over many decades by model theorists, and many structures have been identified in which every partitioned formula is stable: see, for example \citep{classification}. There is an analogous notion of ``continuous logic'', in which formulas are real-valued, and in continuous logic partitioned formulas give real-valued hypothesis classes. The results in works such as \citep{ibycontinuousrandom,itaykeisler} are phrased in terms of NIP and stability of continuous logic structures. 

The notion of measurable family and its expectation class was introduced in \citep{ibycontinuousrandom}.
It relates to a transformation that is similar to the one we perform here: moving from a logical structure to its ``randomization'', another structure where the elements are random variables.
The notion of randomization originates in \citep{keislerrandomizing}, and was developed further in \citep{ibycontinuousrandom,keislerrandomizing,itaykeisler}. Many of our results on both agnostic PAC learning and agnostic online learning of statistical classes can be seen as refinements of results in these works. Our refinements include: translating these results outside of their original context of hypothesis classes coming from logic to general hypothesis classes, and presenting sample complexity bounds. In the agnostic online learning setting we generalize results proven in prior work only for concept classes to the setting of real-valued hypothesis classes.

Model-theoretic characterizations of which partitioned formulas are learnable in other learning models (e.g. Private PAC learning) are provided in \citep{private_pac}.

\section{Conclusions} \label{sec:conc}

We investigated a mapping that takes a ``base'' hypothesis class, consisting of either Boolean or real-valued functions,  to other classes based on probability distributions over either the range space or the parameter space of the class. We connected this to a theory of randomized families of classes, stemming from work in model theory: there we move from a random hypothesis class to its expectation class.
We have proven that these transformations preserve agnostic PAC learnability and agnostic online learnability, refining results from both learning of database queries \citep{sigmod22} and the model theory \citep{itaykeisler}. In addition to providing a linkage between these communities, our results provide improved bounds.  For realizable learning, we have provided counterexamples to preservation.

Our motivation concerns distribution classes, but we obtain our positive results by embedding into a strictly more general setting, the expectation class of a measurable family.  This setting allows correlation between range elements and parameters. We are currently exploring how to exploit this generality.

We leave open the question of whether realizable online learnability of a base class implies realizable online learnability of the distribution class. For the dual distribution class, we showed failure of preservation in Proposition \ref{prop:nonclosurerealizableonline}.

\bibliography{vccl}

\appendix
\onecolumn

\section{Agnostic online learnabilty and duality: Proof of Proposition \ref{prop:onlinedual}} \label{app:dualityagnosticonline}

Recall Proposition \ref{prop:onlinedual}:

\medskip

 A function class is agnostic online learnable exactly when its dual class is.

 \medskip

Although this is surely well known, we include a proof for completeness.  
We can use the characterization in terms of stability of a hypothesis class in Theorem \ref{thm:fatshatteringandstable}:
The theorem shows that a hypothesis class $\hypoclass$ on $\rangespace$ parameterized by $\paramspace$ is online learnable if and only if:

\medskip

For all $\gamma > 0$, there is some $d$ such that
        there are no $a_1,\dots,a_d \in \rangespace$, $b_1,\dots,b_d \in \paramspace$, such that for all $i < j$,
        $$|\hypo_{b_j}(a_i) - \hypo_{b_i}(a_j)| \geq \gamma.$$

Clearly, if the roles of range and parameter space are swapped, the same thing holds. 
\section{Justifying and generalizing Example \ref{ex:deffunction}} \label{app:definablefunctionexample}
Recall that in the body of the paper we presented Example \ref{ex:deffunction}, an example of a hypothesis class of functions.
We claimed that this was agnostic PAC learnable. This can be argued directly via computation of fat-shattering dimensions.
Here we give a more general argument, deriving from logic.

Consider the real numbers with operations addition, multiplication, and the inequality: this is a real-closed order field.
We recall what it means for a set of tuples in the reals to be \emph{first-order definable} over this vocabulary
The \emph{terms} of first-order logic are built up from real constants and variables by applying the functios addition and multiplication: that is, they are polynomials with real coefficients. The \emph{atomic formulas} are inequalities between terms. The \emph{formulas} of first-order logic are built up from atomic formulas via the Boolean operations $\wedge, \vee, \neg$ and the quantifiers $\exists x$, $\forall x$.

A \emph{partitioned formula} $\phi(x_1 \ldots x_j; w_1 \ldots w_\ell)$ is a formula where the free variables are partitioned into two subsets, $\vec x$ and $\vec y$. Such a formula is associated with a family of subsets of $\reals^j$, indexed by $\reals^\ell$,
in the obvious way.  We can similarly talk about a formula with a partition of the free variables into three parts.

A $3$-partitioned formula of the form $\phi(x_1 \ldots x_j; y_1 \ldots y_k; w_1 \ldots w_{\ell})$ is a \emph{definable family of functions} if for every $\vec w$ and $\vec x$, there is exactly one $\vec y$ that makes the formula true.
Such a formula is associated with a family of functions from $\reals^j$ to $\reals^k$, indexed by $\reals^{\ell}$.

Note that the class of Example \ref{ex:deffunction} is of this form.

The following proposition follows from the fact that the real ordered field is an ``NIP structure'':

\begin{proposition} Every definable family of functions over the real ordered field has finite $\gamma$ fat-shattering dimension for each $\gamma>0$, and is thus agnostic PAC learnable.
\end{proposition}

Although this is also probably well known, we sketch a proof:
\begin{proof}
To use the terminology  within model theory, the real-ordered field is o-minimal, and for every o-minimal structure over the reals, it is shown in \cite{vandendries} that every definable family family of subsets in an o-minimal structure has finite  VC dimension. Thus if we have $\phi(\vec x;\vec w)$, with $\phi$ a formula over the real field, the hypothesis class consisting of functions $\hypo_{\vec w}$, the characteristic function  $\phi(\vec x;\vec w)$ has finite $\gamma$ fat-shattering dimension for each $\gamma$. Thus the same holds for finite sums (with fixed number of summands) of such characteristic functions. A definable family of functions, like the family in Example \ref{ex:deffunction}, can be approximated within any tolerance in the sup norm, uniformly in the parameters $\vec w$ by finite sums of characteristic functions. Thus the proposition holds.
\end{proof}
\section{Measurability issues} \label{app:measurability}
In the body of this paper, there are two important steps where our arguments ignored some issues related to measurability assumptions, and two points at which we cite other work without clarifying the measurability assumptions. We spell out these steps and the assumptions here.

Let us first discuss the agnostic PAC case. We have given two arguments for preservation of agnostic PAC learnability to statistical classes in the body of the paper. Both of them proceed by proving bounds on certain combinatorial dimensions, like $\gamma$ fat-shattering. Our story was that this implies finite GC-dimension, which in turn implies agnostic PAC learnability, a standard chain of deductions. However, the argument from combinatorial bounds to GC-dimension bounds in prior work requires some measurability assumptions, which are often not spelled out in detail in the literature: see, e.g. \cite{wirth} for an overview of these issues.

We will not provide exact measurability conditions. But we will show that \emph{countability of the parameter space is sufficient for our sample complexity arguments about the distribution class, and similarly countability of the range space for the arguments about the sample complexity of the dual distribution class.}

Recall that a key point of the argument is
to bound the GC dimensions of the expectation class of a measurable family.
We will argue that:
\begin{itemize} 
\item for this argument it suffices that the parameter space is \emph{separable}: informally, this means we can approximate the behaviour by looking at a countable subset of the parameter space.
\item when we revisit the embedding of the distribution class as a measurable family from Section \ref{subsec:embedding}, it produces a separable measurable family as long as the parameter space of the base class is countable.
\end{itemize}

Similarly, in agnostic online case, our proof of Theorem \ref{thm:randomvaronline} requires additional assumptions, including those required by results cited from \cite{rakhlin2}. We will see that these arguments also work in the context of a separable parameter space.

\subsection{The Necessary Measurability Assumptions}
Recall Lemma \ref{lem:exp_height}:

\medskip

Let $(\Omega,\Sigma,\mu)$ be a probability space, and let $\classfamily = (\hypoclass_\omega : \omega \in \Omega)$ be a  measurable family of hypothesis classes on $\rangespace$.
    
    Fix $n$, $\bar x \in \rangespace^n$, and a Borel probability measure $\beta$ on $\reals^n$. Then 
    $$\width(\expclass\classfamily(\bar x,\paramspace),\beta) \leq 
    \mathbb{E}_\mu[\width(\hypoclass_\omega(\bar x,\paramspace),\beta)].$$

    \medskip
    
The first issue is that for $\mathbb{E}_\mu[\setwidth(\hypoclass_\omega(\bar x,\paramspace),\vec b)]$ in the statement of Lemma \ref{lem:exp_height} to be well-defined, 
we need the measurable family $\classfamily$ to satisfy a stronger measurability assumption:
\begin{definition}
    Assume that $(\Omega,\Sigma,\mu)$ is a probability space.
    Say that a measurable family $\classfamily = (\hypoclass_\omega : \omega \in \Omega)$ of hypothesis classes on $\rangespace$ parameterized by $\paramspace$ is \emph{strongly measurable} when for each $\bar x \in \rangespace^n$  
    and $\vec b \in \reals^n$, the function $\omega \mapsto \setwidth(\hypoclass_\omega(\bar x,\paramspace),\vec b)$ is measurable. 
\end{definition}

Another place measurability is required is defining Glivenko-Cantelli dimensions. The Glivenko-Cantelli dimensions of a hypothesis class $\hypoclass$ are defined in terms of the measures of the sets
$$\left\{(x_1,..., x_m) ~ |~ \exists \hypo \in \hypoclass ,
\left| \frac{1}{m} \cdot (\Sigma^m_{i=1} \hypo(x_i) ) - \int \hypo(u) dD(u)\right| > \epsilon \right\}$$
for $\varepsilon > 0, m$, and a distribution $D$,
which we could rewrite as the union
$$\bigcup_{\hypo \in \hypoclass}\left\{(x_1,..., x_m) ~ |~ 
\left| \frac{1}{m} \cdot (\Sigma^m_{i=1} \hypo(x_i) ) - \int \hypo(u) dD(u)\right| > \epsilon \right\},$$
which in general need not be measurable, although they are if for each $m$, the following function is measurable:
$$(x_1,\dots,x_m) \mapsto \sup_{\hypo \in \hypoclass}\left|\frac{1}{m} \cdot (\Sigma^m_{i=1} \hypo(x_i) ) - \int \hypo(u) dD(u)\right|.$$

In summary, when we refer to the Glivenko-Cantelli dimensions of $\hypoclass$, we implicitly assume that $\hypoclass$ has the following property:
\begin{definition}[well-defined GC dimensions]
    Say a hypothesis class $\hypoclass$ on $X$ \emph{has well-defined Glivenko-Cantelli dimensions} if for all $\varepsilon > 0$, all $m$, and all distributions $D$ on measure space on $X$ induced by $\hypoclass$, the set
    $$\left\{(x_1,..., x_m) ~ |~ \exists \hypo \in \hypoclass ,
\left| \frac{1}{m} \cdot (\Sigma^m_{i=1} \hypo(x_i) ) - \int \hypo(u) dD(u)\right| > \epsilon \right\}$$
    is measurable with respect to the measure space on $X$ induced by $\hypoclass$.
\end{definition}
If $\hypoclass$ has well-defined Glivenko-Cantelli dimensions, then the bounds in Fact \ref{fact:gc_sample} on the learnability of $\hypoclass$ hold.

The two results we cited without measurability assumptions are Facts \ref{fact:rademacher} and \ref{fact:meanwidthonline}:

Fact \ref{fact:rademacher} is used in the proof of Theorem \ref{thm:randomvarclass} to deduce bounds on Glivenko-Cantelli dimensions from Rademacher mean width. The result this rephrases, \cite[Theorem 4.10]{wainwright}, assumes measurability in the form of well-defined Glivenko-Cantelli dimensions. It is observed at the beginning of \cite[Section 4.4]{wainwright} that this measurability assumption holds in the case of a parameter space which is either countable or, in an appropriate metric, separable.

Fact \ref{fact:meanwidthonline} is used in the proof of Theorem \ref{thm:randomvaronline} to deduce online learnability from sequential Rademacher mean width. It is cited from \cite{rakhlin2}, where $\rangespace$ and $\paramspace$ are assumed to be separable metric spaces, with $\hypoclass : X \times Y \to [0,1]$ continuous (or at least, lower-semicontinuous). In the case where $\rangespace$ and $\paramspace$ are countable, this can be satisfied trivially by giving each the discrete metric ($d(a,b) = 1$ when $a \neq b$). We will show that this assumption is also satisfied when $\rangespace$ is countable and $\paramspace$ is a space of random variables on a countable set.

It is worth noting that Theorem \ref{thm:pac_dist_alternate} and Corollary \ref{cor:online_dist_alternate} provide alternate proofs that agnostic PAC and online learning respectively are preserved under passing to the distribution class. The proofs of both of these theorems simply show that finite (sequential) fat-shattering dimension is preserved, from which one can deduce that learnability is as well. However, this latter deduction is once again subject to measurability constraints. In the PAC case, well-defined Glivenko-Cantelli dimensions
suffice, and in the online case, separability suffices.

\subsection{Countable Parameter Spaces}
We now check that if $\paramspace$ is countable, then these measurability assumptions are satisfied.

\begin{lemma}
    If $\classfamily$ is a measurable family of hypothesis classes on $\rangespace$ parameterized by $\paramspace$, and $\paramspace$ is countable, then $\classfamily$ is strongly measurable.
\end{lemma}
\begin{proof}
    Assume $\paramspace$ is countable, and fix $\bar x = (x_1,\dots,x_n) \in \rangespace^n$, $\vec b = (b_1,\dots,b_n) \in \reals^n$. For any $\omega \in \Omega$, we can view $\hypoclass_\omega$ as a function $\hypoclass_\omega : \rangespace \times \paramspace \to [0,1]$. Then expanding definitionally,
    $$\setwidth(\hypoclass_\omega(\bar x,\paramspace),\vec b) = \sup_{y \in \paramspace}\sum_{i = 1}^n\hypoclass_\omega(x_i,y)b_i.$$
    
    The assumption that $\classfamily$ is measurable means that for each $x \in \rangespace,y \in \paramspace$, the function $\omega \mapsto \hypoclass_\omega(x,y)$ is measurable, from which we see that for each $y \in \paramspace$, the function
    $\omega \mapsto \sum_{i = 1}^n\hypoclass_\omega(x_i,y)b_i$ is measurable.
    A countable supremum of measurable functions is measurable, so
    $\omega \mapsto \setwidth(\hypoclass_\omega(\bar x,\paramspace),\vec b)$ is measurable, making $\classfamily$ a strongly measurable class.
\end{proof}

It is also clear that if $\hypoclass$ is a class parameterized by a countable set $\paramspace$, the Glivenko-Cantelli dimensions are well-defined, because the sets in question are countable unions.

We can also conclude that if $\classfamily$ is a measurable family of hypothesis classes parameterized by countable $\paramspace$, then Theorem \ref{thm:randomvarclass} holds. Because a countable set is also separable in any metric, the measurability requirements for Fact \ref{fact:meanwidthonline} are satisfied, so Theorem \ref{thm:randomvaronline} holds.

\subsection{From countable parameter spaces to the parameters of random variable classes}
When we analyze a class $\hypoclass$ under the assumption that its parameter set $\paramspace$ is countable, we eventually have to consider the class $\hypoclass(\compatparamrv(\hypoclass))$, parameterized by $\compatparamrv(\hypoclass)$, which need not be countable. However, we will show that it is separable in a suitable metric, and show that this suffices for the rest of our measurability requirements.

\begin{lemma}\label{lem:rv_sep}
    If $\paramspace$ is countable, then the set $\compatparamrv(\hypoclass)$ is separable with the metric where $$d(\paramspace_1,\paramspace_2) = \sum_{y \in \paramspace} \left|\mathbb{P}[\paramspace_1 = y] - \mathbb{P}[\paramspace_2 = y]\right|.$$
\end{lemma}
\begin{proof}
    Let $A$ be the set of all $Y^* \in \compatparamrv(\hypoclass)$ such that for all $y \in \paramspace,$ $\mathbb{P}[Y^* = y] \in \mathbb{Q}$, for all but finitely many $y \in \paramspace$, $\mathbb{P}[Y^* = y] = 0$. This set is countable, and to show that it is dense, we will show that for all $Y' \in \compatparamrv(\hypoclass)$, and $\varepsilon > 0$, there is some $Y_\varepsilon \in A$ such that $d(Y',Y_\varepsilon) \leq \varepsilon$.

    As $\sum_{y \in \paramspace}\mathbb{P}[Y' = y] = 1$,
    there is some finite set $S \subseteq \paramspace$ with $\sum_{y \in S}\mathbb{P}[Y' = y] \geq 1 - \frac{\varepsilon}{3}$.
    Then by density of the rational numbers, we can find nonnegative rational numbers $(r_y: y \in S)$ such that $\sum_{y \in S}r_y = 1$, and $\sum_{y \in S}|\mathbb{P}[Y' = y] - r_y| \leq \frac{2\varepsilon}{3}$.
    We then define $Y_\varepsilon$ so that if $y \in S$, then
    $\mathbb{P}[Y_\varepsilon = y] = r_y$,
    and otherwise, $\mathbb{P}[Y_\varepsilon = y] = 0$.

    Then $d(Y',Y_\varepsilon) =
    \sum_{y \in S}\sum_{y \in S}|\mathbb{P}[Y' = y] - r_y| + \sum_{y \in \paramspace \setminus S}|\mathbb{P}[Y' = y] - 0| \leq \frac{2\varepsilon}{3} + \frac{\varepsilon}{3} = \varepsilon.$
\end{proof}

We now check that separability in this particular metric suffices to imply well-defined Glivenko-Cantelli dimension, justifying the application of Fact \ref{fact:rademacher} in the proof of Theorem \ref{thm:randomvarclass}.

\begin{theorem}
    If $\paramspace$ is countable, then for any $\delta, \varepsilon > 0$, $\glivcant_{\varepsilon,\delta}\left(
    \hypoclass(\compatparamrv(\hypoclass))\right)$ is well-defined.
\end{theorem}
\begin{proof}
    It suffices to show that for all $m$, the function
    $$(x_1,\dots,x_m) \mapsto \sup_{\paramspace' \in \compatparamrv(\hypoclass)}\left|\frac{1}{m} \cdot (\Sigma^m_{i=1} \mathbb{E}_\mu[\hypo_{\paramspace'}(x_i)] ) - \int \mathbb{E}_\mu[\hypo_{\paramspace'}(u)] dD(u)\right|.$$
    is measurable. To do this, let $A$ be a countable dense set of $\compatparamrv(\hypoclass)$ in the metric from Lemma \ref{lem:rv_sep}, and we will show that the functions $f_{Y'}$ defined by
    $$f_{Y'}(x_1,\dots,x_m) = \left|\frac{1}{m} \cdot (\Sigma^m_{i=1} \mathbb{E}_\mu[\hypo_{\paramspace'}(x_i)] ) - \int \mathbb{E}_\mu[\hypo_{\paramspace'}(u)] dD(u)\right|$$
    are themselves measurable, and that $\{f_{\paramspace'} : \paramspace' \in A\}$ is dense in $\{f_{\paramspace'} : \paramspace' \in \compatparamrv(\hypoclass)\}$, in the $\sup$-metric, so any $f_{\paramspace'}$ is a uniform limit of elements of $\{f_{\paramspace'} : \paramspace' \in A\}$, from which we conclude that
    \begin{align*}
        \sup_{\paramspace' \in \compatparamrv(\hypoclass)}f_{\paramspace'} = \sup_{\paramspace' \in A}f_{\paramspace'},
    \end{align*}
    and the latter countable supremum is measurable.

    First, let $Y' \in \compatparamrv(\hypoclass)$.
    The function $f_{Y'}$ is measurable, because it is the absolute value of a sum of measurable functions, as $x \mapsto \mathbb{E}_\mu[\hypo_{\paramspace'}(x)]$ is the uniform limit of linear combinations of measurable functions $\hypo_y(x)$ for $y \in \paramspace'$.

    Then we show that $\{f_{\paramspace'} : \paramspace' \in A\}$ is dense in $\{f_{\paramspace'} : \paramspace' \in \compatparamrv(\hypoclass)\}$, in the $\sup$-metric, by showing that if $\paramspace_1,\paramspace_2 \in \compatparamrv(\hypoclass)$, then $\sup_{\bar x}|f_{\paramspace_1}(\bar x) - f_{\paramspace_2}(\bar x)| \leq 2d(\paramspace_1,\paramspace_2)$.

    We see that for every $x \in \rangespace$,
    \begin{align*}
        \mathbb{E}_\mu[\left|\hypo_{\paramspace_1}(x) - \hypo_{\paramspace_2}(x)\right|]
        &\leq \sum_{y \in \paramspace}\hypo_y(x)\left|\mathbb{P}[\paramspace_1 = y] - \mathbb{P}[\paramspace_1 = y]\right|\\
    &\leq \sum_{y \in \paramspace}\left|\mathbb{P}[\paramspace_1 = y] - \mathbb{P}[\paramspace_1 = y]\right|\\
    &= d(\paramspace_1,\paramspace_2),\\
    \end{align*}
    so for every $\bar x \in \rangespace^m$,
    \begin{align*}
        &f_{\paramspace_1}(\bar x) - f_{\paramspace_2}(\bar x)\\
    =\,& \left|\frac{1}{m} \cdot (\Sigma^m_{i=1} \mathbb{E}_\mu[\hypo_{\paramspace_1}(x_i)] ) - \int \mathbb{E}_\mu[\hypo_{\paramspace_1}(u)] dD(u)\right| - \left|\frac{1}{m} \cdot (\Sigma^m_{i=1} \mathbb{E}_\mu[\hypo_{\paramspace_2}(x_i)] ) - \int \mathbb{E}_\mu[\hypo_{\paramspace_2}(u)] dD(u)\right|\\
    \leq\,& \frac{1}{m} \cdot (\Sigma^m_{i=1} \mathbb{E}_\mu\left[\left|\hypo_{\paramspace_1}(x_i) - \hypo_{\paramspace_2}(x_i)\right|\right]) + \int \mathbb{E}_\mu\left[\left|\hypo_{\paramspace_1}(u) - \hypo_{\paramspace_2}(u)\right|\right] dD(u)\\
    \leq\,& 2d(\paramspace_1,\paramspace_2).
    \end{align*}
\end{proof}

Furthermore, to show that the assumptions of Fact \ref{fact:meanwidthonline} are satisfied, we check the following:
\begin{lemma}
    Suppose $\hypoclass$ is a hypothesis class with range space $\rangespace$ and parameter space $\paramspace$, which are both countable.
   Then $\hypoclass(\compatparamrv(\hypoclass)) : \rangespace \times \compatparamrv(\hypoclass) \to [0,1]$ is continuous, where $\rangespace$ has the discrete metric and $\compatparamrv(\hypoclass)$ has the above separable metric.
\end{lemma}
\begin{proof}
    Because $\rangespace$ has the discrete metric, it suffices to show that for each $x \in \rangespace$, the function
    $$f_x : Y' \mapsto (\hypoclass(\compatparamrv(\hypoclass)))_{Y'}(x)$$
    is continuous.
    
    In fact, we will see that it is 1-Lipschitz.
    Suppose $Y_1, Y_2 \in \compatparamrv(\hypoclass))$.
    Then
    \begin{align*}
        |f_x(Y_1) - f_x(Y_2)|
    =&\left|\sum_{y \in \paramspace}\mathbb{P}[Y_1 = y]\hypo_y(x) - \sum_{y \in \paramspace}\mathbb{P}[Y_1 = y]\hypo_y(x)\right|\\
    =&\left|\sum_{y \in \paramspace}\left(\mathbb{P}[Y_1 = y] - \mathbb{P}[Y_1 = y]\right)\hypo_y(x)\right|\\
    \leq & \sum_{y \in \paramspace}\left|\mathbb{P}[Y_1 = y] - \mathbb{P}[Y_1 = y]\right|\hypo_y(x)\\
    \leq & \sum_{y \in \paramspace}\left|\mathbb{P}[Y_1 = y] - \mathbb{P}[Y_1 = y]\right|\\
    &= d(Y_1,Y_2).
    \end{align*}
\end{proof}


    

\subsection{Without Countability Assumptions}
If we place no countability or separability assumptions on $\paramspace$,
then we can still prove results which transfer from a base class to a statistical class, but which only deal with combinatorial dimensions, rather than learnability. Here would be one transfer result for the agnostic PAC case:

\begin{theorem}
\label{thm:dim_exp_uncountable}
    If $\classfamily$ is a measurable family of hypothesis classes with uniformly bounded $\gamma$ (sequential) fat-shattering dimension for all $\gamma > 0$, then $\expclass\classfamily$ has finite $\gamma$ (sequential) fat-shattering dimension for all $\gamma > 0$.  
\end{theorem}
\begin{proof}
    We prove this first for fat-shattering dimension, and then consider the sequential version.

    If the class $\expclass\classfamily$ has infinite $\gamma$ fat-shattering dimension for some $\gamma > 0$, this is witnessed by a countable subset $\paramspace' \subseteq \paramspace$.
    If we restrict the parameters of every class in $\classfamily$ to $\paramspace'$, the resulting measurable family $\classfamily'$ is strongly measurable.
    Thus $\expclass\classfamily'$ having infinite $\gamma$ fat-shattering dimension implies, by Theorems \ref{thm:exp_mean_width} and  and the connection between finite fat-shattering dimensions and Rademacher mean width, that there must be some $\delta > 0$ for which the classes $\hypoclass \in \classfamily'$ have unbounded fat-shattering dimension, and the same applies to $\classfamily$. 

    For the sequential version, the argument is the same, but citing Theorem \ref{thm:exp_seq_mean_width}.
\end{proof}

However, Theorems \ref{thm:randomvarclass} and \ref{thm:randomvaronline} are still subject to the measurability assumptions of Facts \ref{fact:rademacher} and \ref{fact:meanwidthonline} respectively.

\section{More detail on Fact \ref{fact:preservestablerandom}} \label{app:translatestability}

Recall that in the body we stated Fact  \ref{fact:preservestablerandom}, about passing from a uniformly stable family to stability of its expectation class.
We cited \cite{itaykeisler,ibyrandvar} for this, and mention that the language used in these works is quite different. We now explain how to translate from the setting of these works to Fact \ref{fact:preservestablerandom}.

Our statement is that if a measurable family $\classfamily$ of hypothesis classes on $\rangespace$ parameterized by $\paramspace$ is uniformly stable, then
the expectation class $\expclass\classfamily$ is also stable.
The statements in \cite{itaykeisler,ibyrandvar} show that stability, as a property of a formula in a metric structure of continuous logic, is preserved under a construction known as ``randomization''.
For the definitions of continuous logic and metric structure, see \cite{ibyrandvar,itaykeisler} and references therein.

Given a measurable family $\classfamily$ as above, we will construct such a metric structure, and check that stability of $\expclass\classfamily$ follows from stability of a corresponding formula in the randomization structure.

In order to state this properly, we need to explain what it means for a formula to be stable in a metric structure or theory of continuous logic. We define this in terms of hypothesis classes to match the rest of this paper, but the resulting definition matches \cite[Definition 7.1]{localstab}.
\begin{definition}\label{defn:stable_formula}
    In any metric structure $\mathcal{M}$, given a formula $\phi(x;y)$ of continuous logic, let $\hypoclass_{\mathcal{M},\phi}$ be the class on the $x$-sort of $\mathcal{M}$ parameterized by the $y$-sort of $\mathcal{M}$ and given by $(\hypoclass_{\mathcal{M},\phi})_y(x) = \phi(x,y)$.

    We say that $\phi(x;y)$ is \emph{stable} in $\mathcal{M}$ when $\hypoclass_{\mathcal{M},\phi}$ is. We say that $\phi(x;y)$ is \emph{stable} in a theory $T$ when it is stable in every model $\mathcal{M} \vDash T$.
\end{definition}

We begin by recalling some definitions and results related to randomization from \cite{ibyrandvar}. The special case of concept classes is worked out in \cite{itaykeisler}.

Given a theory $T$ of continuous logic in the language $\mathcal{L}$, there is a language $\mathcal{L}_R$, known as the \emph{randomization} of $\mathcal{L}$, and an $\mathcal{L}_R$-theory $T_R$ of continuous logic, known as the \emph{randomization} of $T$, such that for every $\mathcal{L}$-formula $\phi(\bar x)$ in the language of $T$, there is a corresponding $\mathcal{L}_R$-formula $\mathbb{E}[\phi(\bar x)]$ in the language of $T_R$.

We defer to \cite{ibyrandvar} for the exact construction of $\mathcal{L}_R$ and $T_R$. Below we summarize their most relevant properties:

\begin{definition}[{See \citet[Definitions 3.4, 3.23]{ibyrandvar}}]
    Let $(\Omega, \Sigma,\mu)$ be a probability space, and let $(\mathcal{M}_\omega : \omega \in \Omega)$ be a family of metric structures in a relational language of continuous logic. We call $(\mathcal{M}_\omega : \omega \in \Omega)$ a \emph{random family of structures}
    when for all relation symbols $R(\bar x)$ in this language, and tuples $\bar a$ of random variables $a : \omega \to \bigcup_\omega \mathcal{M}_\omega$ where $a(\omega) \in \mathcal{M}_\omega$, 
    the function $\omega \mapsto R(\bar a(\omega))$ is measurable.
\end{definition}

\begin{fact}[{\citet[Corollary 3.24]{ibyrandvar}}]\label{fact:cor324iby}
    Let $T$ be a theory of continuous logic, and let $(\mathcal{M}_\omega : \omega \in \Omega)$ be a random family of models of $T$.
    Then there is a metric $\mathcal{L}_R$-structure
    $\mathcal{R} \vDash T_R$, consisting of the random variables $a : \omega \to \bigcup_\omega \mathcal{M}_\omega$ where $a(\omega) \in \mathcal{M}_\omega$, where for every formula $\phi(\bar x)$,
    $$(\mathbb{E}[\phi(\bar a)])^{\mathcal{R}} = \mathbb{E}_\mu\left[\phi(\bar a(\omega)))^{\mathcal{M}_\omega}\right].$$
\end{fact}

\begin{fact}[{\citet[Theorem 4.9]{ibyrandvar}}]\label{fact:iby_rand_stable}
    If $\phi(\bar x)$ is stable in $T$, then $\mathbb{E}[\phi(\bar x)]$ is stable in $T_R$.
\end{fact}

Recall that $\classfamily = (\hypoclass_\omega : \omega \in \Omega)$ consists of hypothesis classes $\hypoclass_\omega$ on $\rangespace$ parameterized by $\paramspace$, where $(\Omega, \Sigma,\mu)$ is an atomless probability space with respect to which each function $\omega \mapsto \hypoclass_\omega(x,y)$ with $x \in \rangespace, y \in \paramspace$ is measurable.

We now construct a measurable family of structures based on $\hypoclass_\omega$. Given $\omega \in \Omega$, let $\mathcal{M}_\omega$ be a metric structure with two sorts: $\rangespace, \paramspace$. To each sort we assign the discrete metric ($d(x,y) = 1$ when $x \neq y$), so all functions on these sorts will be uniformly continuous. Then we make this a structure in a language with one relation symbol, $\phi(x,y)$, where $x$ is a variable of sort $\rangespace$ and $y$ is a variable of sort $\paramspace$.
We interpret this symbol in $\mathcal{M}_\omega$ by $\phi(x,y) = (\hypoclass_\omega)_y(x)$.

Now let $T$ be the shared theory of all the structures $\mathcal{M}_\omega$.
\begin{lemma}
    The partitioned formula $\phi(x;y)$ is stable in the theory $T$.
\end{lemma}
\begin{proof}
     This means that for every $\mathcal{M} \vDash T$, the class $\hypoclass_{\mathcal{M},\phi}$ on the $\rangespace$-sort of $\mathcal{M}$ parameterized by the $\paramspace$-sort of $\mathcal{M}$ and given by $(\hypoclass_{\mathcal{M},\phi})_y(x) = \phi(x,y)$
    is stable.

   Note that the class $\hypoclass_{\mathcal{M}_\omega,\phi}$ is just $\hypoclass_\omega$.
    
    Because $T$ is the common theory of the structures $\mathcal{M}_\omega$, and for each $r < s$, the $(r,s)$-threshold dimensions of the classes $\hypoclass_{\mathcal{M}_\omega,\phi}$ are uniformly bounded, the same bound is implied by some condition proven by the theory $T$, so this same bound also holds for any other $\mathcal{M} \vDash T$, so $\phi(x;y)$ is stable in $T$.
\end{proof}

This implies also that $\mathbb{E}[\phi(x,y)]$ is stable in $T_R$. This allows us to construct stable hypothesis classes from models of $T_R$.
We see that 
there is a model $\mathcal{R} \vDash T_R$ consisting of random variables where for each pair $a,b$ of random variables in the $\rangespace$-sort and $\paramspace$-sort respectively,
$$(\mathbb{E}[\phi(a,b)])^{\mathcal{R}} = \mathbb{E}_\mu\left[\phi(a(\omega),b(\omega)))^{\mathcal{M}_\omega}\right].$$

We now claim that the hypothesis class $\expclass\classfamily$ embeds into the class defined by $\mathbb{E}[\phi(x,y)]$ in this structure, and is thus stable. To see this, note that for every $x \in \rangespace$, $y \in \paramspace$, the constant random variables $a,b$ taking values $x,y$ respectively are elements of the appropriate sorts of $\mathcal{R}$, by the construction in \cite{ibyrandvar}. Thus for these elements of $\mathcal{R}$,
$$(\mathbb{E}[\phi(a,b)])^{\mathcal{R}} = \mathbb{E}_\mu\left[\phi(x,y)^{\mathcal{M}_\omega}\right]
= \mathbb{E}_\mu\left[(\hypoclass_\omega)_y(x)\right]
= \expclass\classfamily_y(x).$$

\end{document}